\documentclass[reqno,11pt]{article}

\usepackage{amssymb,amsmath,amsthm}
\usepackage{tipa}
\usepackage{epic}
\usepackage{amsfonts}
\usepackage{accents}
\usepackage{texdraw}
\usepackage{color}
\usepackage{eucal}
\usepackage{cite}
\usepackage{bm}
\newtheorem{Proposition}{Proposition}

\numberwithin{equation}{section}
\numberwithin{Proposition}{section}

\def \tyb#1{\hbox{\tiny{[{\it{#1}}]}}}
\def \ty#1{\hbox{\tiny{{\it{#1}}}}}

\def \h#1{\widehat{#1}}
\def \t#1{\widetilde{#1}}

\def \th#1{\widehat{\widetilde{#1}}}

\def \S#1{S^{(#1)}}

\DeclareMathAccent{\wtilde}{\mathord}{largesymbols}{"65}
\DeclareMathAccent{\what}{\mathord}{largesymbols}{"62}

\def\m@th{\mathsurround=0pt}
\mathchardef\bracell="0365
\def\upbrall{$\m@th\bracell$}
\def\undertilde#1{\mathop{\vtop{\ialign{##\crcr
    $\hfil\displaystyle{#1}\hfil$\crcr
     \noalign
     {\kern1.5pt\nointerlineskip}
     \upbrall\crcr\noalign{\kern1pt
   }}}}\limits}

\def\underhat#1{\mathop{\vtop{\ialign{##\crcr
    $\hfil\displaystyle{#1}\hfil$\crcr
     \noalign
     {\kern1.5pt\nointerlineskip}
     \upbrall\crcr\noalign{\kern1pt
   }}}}\limits}

\newcommand{\wh}{\widehat}
\newcommand{\wt}{\widetilde}

\newcommand{\bA}{\boldsymbol{A}}
\newcommand{\bB}{\boldsymbol{B}}
\newcommand{\bC}{\boldsymbol{C}}
\newcommand{\bI}{\boldsymbol{I}}
\newcommand{\bk}{\boldsymbol{k}}
\newcommand{\bK}{\boldsymbol{K}}

\newcommand{\bM}{\boldsymbol{M}}
\newcommand{\bS}{\boldsymbol{S}}
\newcommand{\bc}{\boldsymbol{c}}
\newcommand{\br}{\boldsymbol{r}}
\newcommand{\bs}{\boldsymbol{s}}
\newcommand{\bU}{\boldsymbol{U}}

\newcommand{\mbe}{\mathbf{e}}
\newcommand{\bu}{\boldsymbol{u}}

\newcommand{\bT}{\boldsymbol{T}}
\newcommand{\bF}{\boldsymbol{F}}
\newcommand{\bG}{\boldsymbol{G}}
\newcommand{\bH}{\boldsymbol{H}}
\newcommand{\Ga}{\boldsymbol{\Gamma}}
\newcommand{\bW}{\boldsymbol{W}}
\newcommand{\bY}{\boldsymbol{Y}}
\newcommand{\ST}{\hbox{\tiny\it{T}}}

\begin{document}

\title{The Sylvester equation and the elliptic Korteweg-de Vries system}

\author{Ying-ying Sun$^{1}$, ~~ Da-jun Zhang$^{1}$\footnote{Corresponding author. Email: djzhang@staff.shu.edu.cn}, ~~ Frank W. Nijhoff$^{2}$ \\
{\small \it $^{1}$Department of Mathematics, Shanghai University, Shanghai 200444, P.R. China}\\
{\small \it $^{2}$Department of Applied Mathematics, School of Mathematics, University of Leeds, LS2 9JT, UK}}

\maketitle

\begin{abstract}

The elliptic Korteweg-de Vries (KdV) system is a multi-component
generalization of the lattice potential KdV equation, whose soliton
solutions are associated with an elliptic Cauchy kernel (i.e., a Cauchy kernel on the torus).
In this paper we generalize the class of solutions by using a Sylvester type
matrix equation and rederiving  the system from the associated Cauchy matrix.
Our starting point is the Sylvester equation in the form of
$~\boldsymbol{k} \boldsymbol{M}+ \boldsymbol{M} \boldsymbol{k}
= \boldsymbol{r} {\boldsymbol{c}}^{T}-g\boldsymbol{K}^{-1} \boldsymbol{r} {\boldsymbol{c}}^{T} \boldsymbol{K}^{-1}$
where $\boldsymbol{k}$ and $\boldsymbol{K}$ are commutative matrices and obey the matrix relation
${\boldsymbol{k}}^2=\boldsymbol{K}+3e_1\boldsymbol{I}+g{\boldsymbol{K}}^{-1}$.
The obtained elliptic equations, both discrete and continuous, are formulated by the scalar function $S^{(i,j)}$
which is defined using $(\boldsymbol{k},\boldsymbol{K}, \boldsymbol{M}, \boldsymbol{r},\boldsymbol{c})$ and constitute an infinite size symmetric matrix.
Lax pairs for both the discrete and continuous system are derived.
The explicit solution $\boldsymbol{M}$ of the Sylvester equation and generalized solutions of the obtained elliptic equations
are presented according to the canonical forms of matrix $\boldsymbol{k}$.




\end{abstract}

\maketitle

\section{Introduction}
\label{sec-1}

The following elliptic lattice potential Korteweg-de Vries KdV (elpKdV) system
\begin{subequations}\label{ell-la-kdv}
\begin{align}
&(a+b+u-\wh{\wt{u}})(a-b+\wh{u}-\wt{u})
=a^2-b^2+g(\wt s-\wh s)(\wh{\wt{s}}-s),\label{ell-la-kdv-1}\\
&(\wh{\wt{s}}-s)(\wt{w}-\wh{w})=
[(a+u)\wt s-(b+u)\wh s]\wh{\wt{s}}-[(a-\wh{\wt{u}})\wh s-(b-\wh{\wt{u}})\wt s]s,\label{ell-la-kdv-2}\\
&(\wh{s}-\wt s)(\wh{\wt{w}}-{w})=
[(a-\wt u) s+(b+\wt u)\wh{\wt{s}}]\wh{s}-[(a+\wh{u})\wh{\wt{s}}+(b-\wh{u}) s]\wt s, \label{ell-la-kdv-3}\\
&(a+ u-\frac{\wt w}{\wt s})({a-\wt u+\frac{w}{s}})=a^2-R(s\t s), \label{ell-la-kdv-4}\\
&(b+ u-\frac{\wh w}{\wh s})({b-\wh u+\frac{w}{s}})=b^2-R(s\h s)\label{ell-la-kdv-5}
\end{align}
\end{subequations}
and the continuous elliptic potential Korteweg-de Vries KdV (epKdV) system
\begin{subequations}
\label{c-ell-KdV-couple-sec1}
\begin{align}
s_t=&4s_{xxx}+6s_x[R(s^2)-A^2-\frac{2A s_x}{s}-\frac{2s_{xx}}{s}], \label{c-ell-KdV-01}\\
A_t=&4A_{xxx}-6A^2A_x +6A_xR(s^2)-\frac{6 s_x}{s}(R(s^2))_x, \label{c-ell-KdV-02}
\end{align}
\end{subequations}
with $A=-u+\frac{w}{s}$ were first derived in \cite{NP-JNMP-2003} through the direct linearisation approach.
Here $R(x)$ is associated with the elliptic curve
\begin{equation}\label{ell-curve-2}
    y^2=R(x)=\frac{1}{x}+3e_1+gx,
\end{equation}
where $e_1, g \in \mathbb{C}$ are moduli of the elliptic curve. In equation \eqref{ell-la-kdv} we use
the conventional tilde-hat notations to express shifts w.r.t. discrete variables,
e.g.,
\begin{equation}
 u\doteq u(n,m) \doteq u_{n,m},~ \wt u \doteq u_{n+1,m}, ~\wh u \doteq u_{n,m+1},~\th u \doteq u_{n+1,m+1}.
\end{equation}
The direct linearisation approach starts from  an integral over certain  set of contours,
from which an infinite order matrix $\bU$ is introduced.
Then the closed forms of the elements of $\bU$ yield nonlinear lattice
equations (as examples, see \cite{NQC-PLA-1983,NP-JNMP-2003,N-arxiv-2011,ZZN-SAM-2012}).

In this paper, we will rederive the above two elliptic systems
using the Cauchy matrix approach,
which was successfully applied in \cite{N-2004-math,NAH-2009-JPA} to derive integrable lattice
equations and to analyse  their underlying structures.
The direct linearisation approach and the Cauchy matrix approach
look different w.r.t. their individual procedures,
but in fact they are deeply related together.
The former needs a Cauchy kernel as a key auxiliary role
and both approaches make use of a same infinite order matrix $\bU$ of which
elements generate scaler nonlinear equations.
We use the notation $\bS$ for this infinite order matrix in our paper.

The Cauchy matrix approach is a pure algebraic procedure and it enables us
to obtain equations, their explicit soliton solutions and Lax pairs.
In the Cauchy matrix approach, the Sylvester equation
\begin{equation}
\bA\bM-\bM\bB=\bC
\label{SE}
\end{equation}
can be viewed as a starting point\cite{XZZ-2014-JNMP,ZZ-SAM-2013}. The matrix $\bM$ is a dressed Cauchy matrix (see
the factorization \eqref{factor-M})
and is used to introduce $\tau$-function
\cite{N-2004-math,NAH-2009-JPA,XZZ-2014-JNMP}.

In the present paper we start from the following Sylvester equation:
\begin{equation}
 \bk \bM+ \bM \bk= \br{\bc}^{\ST}-g\bK^{-1}\br {\bc}^{\ST} \bK^{-1},
\label{SE-2a}
\end{equation}
where $\br=(r_1,r_2,\cdots,r_N)^{\ST}$, $\bc=(c_1,c_2,\cdots,c_N)^{T}$
and $\bk,~\bK\in \mathbb{C}_{N\times N}$
obey the matrix relation
\begin{equation}\label{ell-relation-2}
    {\bk}^2=\bK+3e_1\bI+g{\bK}^{-1},~~ \bk\bK=\bK\bk,
\end{equation}
in which $\bI$ is the $N$th-order unit matrix.
Based on the above Sylvester equation, the dispersion relations for the elpKdV system is defined by
\begin{equation}
(a \bI- \bk)\wt{\br}  =(a \bI+\bk) \br,~~
(b \bI-\bk)\wh{\br}  =(b \bI+\bk) \br,
\end{equation}
and for the epKdV system by
\begin{equation}
\br_x=\bk \br,~~\br_t=4\bk^3 \br,~~
\bc^{\ST}_x=\bc^{\ST}\bk,~~\bc^{\ST}_t=4\bc^{\ST} \bk^3.
\end{equation}

In next section, we will focus on the matrix system \eqref{ell-relation-2}.
We will see that the matrix system \eqref{ell-relation-2} is actually governed by the points on the elliptic curve \eqref{ell-curve-2}.
Sec.\ref{sec-3} will focus on the Sylvester equation \eqref{SE-2a}
and the scalar function $\S{i,j}$ defined in \eqref{Sij}.
Explicit solution $\bM$ of \eqref{SE-2a} will be given for the cases of $\bk$ being diagonal and Jordan block form and their combination.
Solutions of the elliptic elpKdV and epKdV equations are consequently obtained.
As a generic element, $\S{i,j}$ composes an infinite order matrix $\bS$ (same as the matrix $\bU$ in the direct linearisation approach
in \cite{NP-JNMP-2003}).
$\{\S{i,j}\}$ satisfy some recurrence relations which can be viewed as discrete equations of $\S{i,j}$ defined on $\mathbb{Z}\times \mathbb{Z}$
and will play crucial roles in deriving the continuous epKdV system.

Apart from the Sec.\ref{sec-2} and \ref{sec-3} above mentioned,
in Sec.\ref{sec-4} and Sec.\ref{sec-5} we respectively derive the elpKdV system and epKdV system
together with their Lax pairs.
In Sec.\ref{sec-6} we discuss continuum limits of the elpKdV system.
Sec. \ref{sec-7} is for conclusions.
Besides, in Appendix we list properties of lower triangular Toeplitz matrices
which play important roles in our paper.

\section{Points on the elliptic curve: parametrization and selection }\label{sec-2}

\subsection{Scalar case}\label{sec-2-1}

Consider the elliptic curve
\begin{equation}
k^2=R\Bigl(\frac{1}{K}\biggr)=K+3e_1+\frac{g}{K}.
\label{ec-sca}
\end{equation}
The discrete plane wave factor is defined as
\begin{equation}
\rho_i=\biggl(\frac{a+k_i}{a-k_i}\biggr)^n \biggl(\frac{b+k_i}{b-k_i}\biggr)^m\,\rho^0_{i}
\label{(2.2)}
\end{equation}
with a phase factor $\rho^0_i$, where we require that $k_i$ together with $K_i$ obeys the elliptic curve \eqref{ec-sca}, i.e.
\begin{equation}
k^2_i=K_i+3e_1+\frac{g}{K_i}.
\label{ec-sca-i}
\end{equation}
In the case of classical soliton solutions (see \cite{NAH-2009-JPA,HZ-JPA-2009}),
$\{k_i\}$ play the role of wave numbers, which should be distinct so that they can represent different solitons.
Since $k_i$ and $K_i$ are coupled through the elliptic curve \eqref{ec-sca-i},
(say, both $k_i$ and $-k_i$ correspond to a same $K_i$,)
we consequently require that they can identify each other, i.e.
\begin{equation}
k_i\neq k_j \Leftrightarrow K_i\neq K_j.
\label{id-cond-1}
\end{equation}
Note that for the arbitrary two points $(k_i,K_i)$ and  $(k_j,K_j)$ on the elliptic curve \eqref{ec-sca} we always have the relation
\begin{equation}
(k_i+k_j)(k_i-k_j)=(K_i-K_j)\,\frac{K_iK_j-g}{K_iK_j}.
\label{fact-ec}
\end{equation}
This means that if we take
\begin{equation}
(k_i+k_j)(K_iK_j-g)\neq 0,
\label{id-cond-2}
\end{equation}
then \eqref{id-cond-1} is guaranteed.
Equation \eqref{id-cond-2} is  the criteria that we select points from the elliptic curve \eqref{ec-sca}.
As a consequence of \eqref{id-cond-2}, $k_i\neq 0$.

The elliptic curve \eqref{ec-sca} can be parameterized using Weierstrass's elliptic function $\wp(\kappa)$ as the following (cf.\cite{NP-JNMP-2003}):
\begin{subequations}
\begin{align}
& K=\wp(\kappa)-e_1,~~ k=\frac{\wp'(\kappa)}{2(\wp(\kappa)-e_1)},\label{param-kK}\\
& e_1=\wp(\omega),~~ g=(e_1-e_2)(e_1-e_3),\label{param-eg}
\end{align}
\label{param}
\end{subequations}
where $e_2=\wp(\omega+\omega'),~e_3=\wp(\omega')$ and $\omega$ and $\omega'$ are respectively the half periods of $\wp(\kappa)$.

An elliptic function is a meromorphic function with double (complex) periods.
Suppose that $f(\kappa)$ is an elliptic function and its double  periods are $2\omega$ and $2\omega'$, respectively.
Let $\mathbb{D}$ denote a fundamental  period parallelogram $ABCD$ as described in  Fig.\ref{F:1}.
The border segments $AB$ and $BC$ together with points $\{A,B,C\}$ are not included in $\mathbb{D}$
due to the double periodicity.
\begin{figure}[h]
\setlength{\unitlength}{0.0004in}
\vskip 0.5cm
\hspace{4.5cm}
\begin{picture}(3482,2813)(0,-10)
\put(1675,2708){\circle{150}}
\put(1075,2908){\makebox(0,0)[lb]{\footnotesize{$-\omega+\omega'$}}} \put(1275,2608){\makebox(0,0)[lb]{\footnotesize{$A$}}}
\put(3275,2708){\circle*{150}}
\put(3175,2908){\makebox(0,0)[lb]{\footnotesize{$\omega'$}}}
\put(4875,2708){\circle*{150}}
\put(4575,2908){\makebox(0,0)[lb]{\footnotesize{$\omega+\omega'$}}}
\put(1275,2608){\makebox(0,0)[lb]{\footnotesize{$A$}}}
\put(4975,2608){\makebox(0,0)[lb]{\footnotesize{$D$}}}

\put(675,1708){\makebox(0,0)[lb]{\footnotesize{$-\omega$}}}
\put(2875,1408){\makebox(0,0)[lb]{\footnotesize{$O$}}}
\put(4475,1708){\circle*{150}}
\put(4675,1708){\makebox(0,0)[lb]{\footnotesize{$\omega$}}}

\put(875,708){\circle{150}}
\put(175,308){\makebox(0,0)[lb]{\footnotesize{$-\omega-\omega'$}}}
\put(2175,308){\makebox(0,0)[lb]{\footnotesize{$-\omega'$}}}
\put(4075,708){\circle{150}}
\put(3675,308){\makebox(0,0)[lb]{\footnotesize{$\omega-\omega'$}}}
\put(475,628){\makebox(0,0)[lb]{\footnotesize{$B$}}}
\put(4275,628){\makebox(0,0)[lb]{\footnotesize{$C$}}}

\drawline(1675,2708)(4875,2708)
\drawline(1275,1708)(4475,1708)
\dashline{150}(875,708)(4075,708)
\dashline{150.000}(1275,1708)(875,708)
\dashline{150.000}(1275,1708)(1675,2708)
\drawline(4875,2708)(4075,708)
\drawline(3275,2708)(2475,708)
\end{picture}
\vskip -0.5cm
\caption{Fundamental period parallelogram $\mathbb{D}$}\label{F:1}
\end{figure}
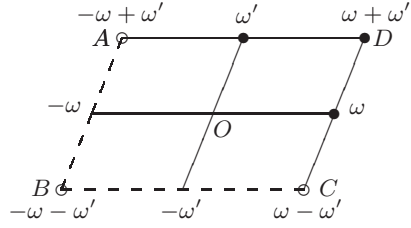
The elliptic function $f(\kappa)$ obeys the Liouville's theorems (cf.\cite{A-book}) that is, for
the elliptic function $f(\kappa)$, in the period parallelogram $\mathbb{D}$
the times that it takes value $a$ equals to the order of $f(\kappa)$.

Under the  parametrization \eqref{param}, for the points on the curve \eqref{ec-sca} we have
\begin{equation}
K_i=\wp(\kappa_i)-e_1,~~ k_i=\frac{\wp'(\kappa_i)}{2(\wp(\kappa_i)-e_1)}.
\label{param-i}
\end{equation}
Then  the criteria \eqref{id-cond-2} can alternatively be
described through the following requirement on $\kappa_i$:
\begin{subequations}
\begin{align}
& \kappa_i\in \mathbb{D}'=\mathbb{D}\setminus\{0,\omega,\omega',\omega+\omega'\},\label{id-cond-3a}\\
& (\wp(\kappa_i)-e_1)(\wp(\kappa_j)-e_1)\neq g,\label{id-cond-3b}\\
& \kappa_i+\kappa_j\neq 0,
\label{id-cond-3c}
\end{align}
\label{id-cond-3}
for $i,j=1,2,\cdots,N$.
\end{subequations}
We note that here and after when we talk about $\kappa_i+\kappa_j$
we always mean $\kappa_i+\kappa_j$ modulo the periodic lattice, i.e. $\kappa_i+\kappa_j (\mathrm{mod}(2\omega,2\omega'))$.
In fact, $0$ is the pole of $\wp(\kappa)$ and thus we require $\kappa_i\neq 0$ to avoid singularities.
Besides, $k_i$ can not be zero as a consequence of $k_i+k_j\neq 0$.
In the light of the Liouville's theorems (cf.\cite{A-book})
since $\wp'(\kappa)$ is a third-order elliptic function
in the period parallelogram $\mathbb{D}$ it only has 3 zeros which are $\omega, \omega'$ and $\omega+\omega'$.
To avoid breaking the one-to-one correspondence of $k_i$ and $K_i$, we
require $\kappa_i\notin \{\omega,\omega',\omega+\omega'\}$. Thus we have \eqref{id-cond-3a}.
\eqref{id-cond-3b} is from the requirement $K_iK_j\neq g$.
For the third one in \eqref{id-cond-3}, we can prove that under
\eqref{id-cond-3a} and \eqref{id-cond-3b} we have
\begin{equation}
\kappa_i+\kappa_j = 0 ~\Leftrightarrow k_i+k_j = 0.
\end{equation}
In fact, since $\wp(\kappa)$ is even and $\wp'(\kappa)$ is odd, from the parametrization \eqref{param-i}
we immediately find that if $\kappa_i+\kappa_j = 0$ then $k_i+k_j=0$.
On the other hand, under \eqref{id-cond-3b}, from the factorization \eqref{fact-ec},
if $k_i+k_j = 0$ and \eqref{id-cond-3b} holds, there must be $K_i=K_j$, which means $\kappa_i=\pm\kappa_j$ in $\mathbb{D}$ in the light of
Liouville's theorems (cf. \cite{A-book}).  The case $\kappa_i=\kappa_j$ is impossible because
this case yields $k_i=k_j=0$ due to $k_i+k_j=0$. $k_i=0$ requires $\wp'(\kappa_i)=0$
which is impossible in $\mathbb{D}'$. Thus, $\kappa_i=-\kappa_j$ is the only choice,
i.e. $\kappa_i+\kappa_j = 0 $.

\subsection{Matrix system \eqref{ell-relation-2}}
\label{sec:2.2}

Let us come to the matrix relation \eqref{ell-relation-2}.
Although at first glance \eqref{ell-relation-2} suggests an interpretation
of this matrix relation as a matrix version of an elliptic curve, it just represents the coordinatization of a collection
of points on the given elliptic curve \eqref{ec-sca}.

To understand this, let us consider a similarity transformation
\begin{equation}
\bk_1=\bT \bk \bT^{-1}, ~~\bK_1=\bT \bK \bT^{-1},
\label{trans-sim}
\end{equation}
where $\bT$ serves as the transform matrix.
Obviously, under the above similarity transformation,
\eqref{ell-relation-2} is formally invariant:
\begin{equation}\label{ell-relation-2-1}
    {\bk_1}^2=\bK_1+3e_1\bI+g{\bK_1}^{-1},~~ \bk_1\bK_1=\bK_1\bk_1.
\end{equation}
Thus, in the following we only need to consider the relation
\begin{equation}\label{ell-relation-2-2}
    {\Ga}^2=\bK+3e_1\bI+g{\bK}^{-1},~~ \Ga\bK=\bK\Ga
\end{equation}
where $\Ga$ is the canonical form of $\bk$.

When
\begin{subequations}
\begin{equation}
\Ga=\mathrm{Diag}(k_1,k_2,\cdots,k_N),
\end{equation}
$\bK$ is taken as
\begin{equation}
\bK=\mathrm{Diag}(K_1,K_2,\cdots,K_N),
\end{equation}
\end{subequations}
where $(k_i,K_i)$ are the points on \eqref{ec-sca}, i.e. satisfying \eqref{ec-sca-i}.
$\{k_i\}$ is the eigenvalue set of $\bk$. Here we request that each $k_i\neq 0$ and $k_i^2\neq k_j^2$ for $i\neq j$.
Under such a requirement one can see that the criteria \eqref{id-cond-2} is satisfied in the light of the factorization \eqref{fact-ec}.

It is interesting to consider the case that $\Ga$ is a $N$-th order Jordan block
\begin{align}
&\Ga
=\left(\begin{array}{cccccc}
k_1 & 0    & 0   & \cdots & 0   & 0 \\
1   & k_1  & 0   & \cdots & 0   & 0 \\
0   & 1  & k_1  & \cdots & 0   & 0 \\
\vdots &\vdots &\vdots &\vdots &\vdots &\vdots \\
0   & 0    & 0   & \cdots & 1   & k_1
\end{array}\right),~~k_1\neq 0.
\label{J}
\end{align}
In this case, $k_1$ is the only eigenvalue of $\bk$ with algebraic multiplicity $N$
(and geometric multiplicity 1).
In terms of the parametrization \eqref{param} we need $\kappa_1$ satisfies
\begin{align}
& \kappa_1\in \mathbb{D}\setminus\{0,\omega,\omega',\omega+\omega'\},~~
\wp(\kappa_1)\neq \sqrt{g}+e_1.
\end{align}
To find $\bK$ that corresponds to the Jordan block \eqref{J}, we will make use of properties of Toeplitz matrices.
Let us shortly introduce such matrices. For more details please see Appendix \ref{A:1}.
A matrix of the following form
\begin{equation}
\bT^{\tyb{N}}(\{a_j\}^{N}_{1})
=\left(\begin{array}{cccccc}
a_1 & 0    & 0   & \cdots & 0   & 0 \\
a_2 & a_1  & 0   & \cdots & 0   & 0 \\
a_3 & a_2  & a_1 & \cdots & 0   & 0 \\
\vdots &\vdots &\cdots &\vdots &\vdots &\vdots \\
a_{N} & a_{N-1} & a_{N-2}  & \cdots &  a_2   & a_1
\end{array}\right)
\label{T}
\end{equation}
is called a lower triangular Toeplitz (LTT) matrix.
All the  $N$-th order LTT matrices compose a commutative set $\mathcal{T}^{\tyb{N}}$.
The complex value matrix \eqref{T} can be generated by an analytic function $f(k)$ at certain point $k=k_0$
through
\begin{equation}
a_j=\frac{\partial^{j-1}_k}{(j-1)!}f(k)|_{k=k_0},~~~(j=1,2,\cdots,N).
\label{aj-f(k)}
\end{equation}
With this correspondence we denote \eqref{T} by $\bT^{\tyb{N}}[f(k_0)]$,
and we have  properties (See proposition \ref{prop-A-3b})
\begin{align*}
(\bT^{\tyb{N}}[f(k_0)])^2=\bT^{\tyb{N}}[f^2(k_0)],~~ (\bT^{\tyb{N}}[f(k_0)])^{-1}=\bT^{\tyb{N}}[1/f(k_0)].
\end{align*}
We also note that the Jordan block \eqref{J} is a LLT matrix generated by $f(k)=k$ at $k=k_1$.

Now we can look for a LTT matrix $\bK$ which  satisfies \eqref{ell-relation-2-2} with $\Ga$  \eqref{J}.
In fact, according to proposition \ref{prop-A-1}, for the Jordan block \eqref{J}, when $\Ga \bK=\bK \Ga$
there must be $\bK \in \mathcal{T}^{\tyb{N}}$.
By taking derivatives with respect to $k$ of the elliptic curve \eqref{ec-sca} at the point $(k_1,K_1)$
where $K_1=K(k_1)$ and $K(k)$ is viewed as an implicit function of $k$ determined by the curve \eqref{ec-sca}, we find (for $i\geq j$)
\begin{equation}\label{de-j}
    \frac{1}{j!}\partial_{k}^{j}\, k^{2}=\frac{1}{j!}\partial_{k}^{j} K(k)+3e_1\delta_{j,0}+\frac{g }{j!}\partial_{k}^{j}\frac{1}{K(k)},
    ~~ (\mathrm{at} ~ k=k_1).
\end{equation}
The l.h.s. and the first term and third term on the r.h.s.
respectively correspond to the elements of the LTT matrices generated by
$f(k)=k^2$, $K(k)$ and $1/K(k)$ at $k=k_1$.
This means for the Jordan block case \eqref{J} if we take $\bK=\bT^{\tyb{N}}[K(k_1)]$ then the relation \eqref{ell-relation-2-2} holds.

In Sec.\ref{sec:3.2.3} we will list solutions of \eqref{ell-relation-2-2} using unified notations.

\section{The Sylvester equation and infinite matrix structure}
\label{sec-3}

In this section we will first investigate solutions of the Sylvester equation \eqref{SE-2a}
and derive explicit expression of $\bM$.
Then, with the help of some special matrices we will investigate
recurrence relations of scalar function $\S{i,j}$ (defined in \eqref{Sij})
and properties of the infinite matrix $\bS$ composed by $\S{i,j}$.

\subsection{Solvability of \eqref{SE-2a}}

For the solution of the Sylvester equation \eqref{SE}, there is the following  well known result \cite{Sylvester-1884}.

\begin{Proposition}\label{prop-1}
Denote the eigenvalue sets of $\bA$ and $\bB$ by $\mathcal{E}(\bA)$ and $\mathcal{E}(\bB)$, respectively.
For the known $\bA, \bB$ and $\bC$, the Sylvester equation \eqref{SE} has a unique solution $\bM$ if and only if
$\mathcal{E}(\bA)\cap \mathcal{E}(\bB)=\varnothing$.
\end{Proposition}

Based on this proposition, we find the following.

\begin{Proposition}\label{prop-1-1}
Consider the Sylvester equation \eqref{SE-2a}, i.e.
\begin{equation}
 \bk \bM+ \bM \bk= \br{\bc}^{\ST}-g\bK^{-1}\br {\bc}^{\ST} \bK^{-1},
\label{SE-3a}
\end{equation}
where  the matrices $\bk$ and $\bK$ satisfy
\begin{subequations}
\begin{align}
&\mathcal{E}(\bk)\cap\mathcal{E}(-\bk)=\varnothing,\label{cond-kK-a}\\
& \mathcal{E}(g\bK^{-1})\cap\mathcal{E}(\bK)=\varnothing,
\label{cond-kK-b}
\end{align}
\label{cond-kK}
\end{subequations}
and the matrix relation \eqref{ell-relation-2}, i.e.
\begin{equation}
{\bk}^2=\bK+3e_1\bI+g{\bK}^{-1},~~ \bk \bK=\bK\bk.
\label{ec-mat}
\end{equation}
Then, the `dual' matrix equation
\begin{equation}
 \bK \bM- \bM \bK= \bk\,\br{\bc}^{\ST}-\br {\bc}^{\ST} \bk
\label{SE-3b}
\end{equation}
holds.
\end{Proposition}

\begin{proof}
Here we note that the condition \eqref{cond-kK-a} is necessary to guarantee the solvability of equation \eqref{SE-3a}
in light of proposition \ref{prop-1}.
Then, left multiplying $\bk$ to \eqref{SE-3a} yields
\[\bk^2 \bM +\bk \bM \bk= \bk  (\br{\bc}^{\ST}-g\bK^{-1}\br {\bc}^{\ST} \bK^{-1}).\]
Next, using the Sylvester equation \eqref{SE-3a} and replacing the term $\bk
\bM$ with $-\bM \bk+\br{\bc}^{\ST}-g\bK^{-1}\br {\bc}^{\ST} \bK^{-1}$,
one has
\begin{equation}\label{s-2}
\bk^2 \bM - \bM \bk^2 = - \br {\bc}^{\ST}\bk +\bk \br {\bc}^{\ST}
+g\bK^{-1}\br {\bc}^{\ST} \bK^{-1}\bk-g\bk\bK^{-1}\br {\bc}^{\ST} \bK^{-1}.
\end{equation}
Making use of the relation \eqref{ec-mat} to replace $\bk^2$ we find
\begin{equation}
   g\bK^{-1}(\bK \bM-\bM \bK-\bk \br{\bc}^{\ST}+\br{\bc}^{\ST} \bk)\bK^{-1} =\bK \bM-\bM \bK-\bk \br{\bc}^{\ST}+\br{\bc}^{\ST} \bk.
\end{equation}
This can be rewritten as a Sylvester equation
\begin{equation}\label{se-W-1}
    g\bK^{-1}\bW-\bW\bK=0,~~~
    \bW=\bK \bM-\bM \bK-\bk \br{\bc}^{\ST}+\br{\bc}^{\ST} \bk.
\end{equation}
Based on proposition \ref{prop-1} and noting that
$\mathcal{E}(g\bK^{-1})\cap\mathcal{E}(\bK)=\varnothing$ in \eqref{cond-kK-b},
the equation \eqref{se-W-1} has a unique solution $\bW=0$, which means
\eqref{SE-3b} holds.
\end{proof}

Here we note that the condition \eqref{cond-kK} is natural because it is actually the criteria \eqref{id-cond-2}
for selecting points from the elliptic curve \eqref{ec-sca}.
We also note that we cannot derive \eqref{SE-3a} from \eqref{SE-3b}.
In fact,  we start from \eqref{SE-3b}, replace $\bK$ using \eqref{ec-mat} and we get
\[\bk^2\bM-\bM \bk^2-g(\bK^{-1}\bM-\bM\bK^{-1})=\bk \br{\bc}^{\ST}-\br{\bc}^{\ST} \bk.
\]
At meantime, from \eqref{SE-3b} we also have
\[\bK^{-1}\bM-\bM\bK^{-1}=-\bK^{-1}\bk \br{\bc}^{\ST}\bK^{-1}+ \bK^{-1}\br{\bc}^{\ST} \bk \bK^{-1}.
\]
Combining them together and using the commutative relation $\bk\bK=\bK\bk$ we reach
\[\bk\bY-\bY\bk=0,~~\bY=\bk \bM+ \bM \bk- \br{\bc}^{\ST}+g\bK^{-1}\br {\bc}^{\ST} \bK^{-1}.\]
Obviously, $\bY=0$  is a solution to the above equation but it is not unique.
This means \eqref{SE-3a} and \eqref{SE-3b} are not equivalent.
\eqref{SE-3a} is more general and \eqref{SE-3b} is a by-product of the former.
In the following discussion it is sufficient that we only consider \eqref{SE-3a}.

\subsection{Solution to the Sylvester equation \eqref{SE-3a}}

\subsubsection{Canonical form of  \eqref{SE-3a}}\label{sec:3.2.1}

Using the similarity transformation \eqref{trans-sim} and denoting
\begin{equation}
\bM_1=\bT \bM \bT^{-1},~~\br_1=\bT
\br,~~{\bc}^{\ST}_1={\bc}^{\ST} \bT^{-1}, \label{Mrs-1}
\end{equation}
it  follows from \eqref{SE-3a} that
\[\bM_1 \bk_1 + \bk_1\bM_1= \br_1 \,{\bc}^{\ST}_1-g\bK^{-1}_1\br_1 {\bc}^{\ST}_1 \bK^{-1}_1,
\]
which is the same form as \eqref{SE-3a}.
This means when we solve the Sylvester equation  \eqref{SE-3a} we only need to consider the following canonical form
\begin{subequations}
\label{SE-cano}
\begin{equation}
\bM \Ga+ \Ga \bM = \br {\bc}^{\ST}-g\bK^{-1}\br {\bc}^{\ST} \bK^{-1},\label{SE-cano-1}
\end{equation}
together with \eqref{ell-relation-2-2}, i.e.
\begin{equation}
{\Ga}^2=\bK+3e_1\bI+g{\bK}^{-1},~~~~\Ga\bK=\bK\Ga , \label{SE-cano-2}
\end{equation}
where
\begin{equation}\br=(r_1,r_2,\cdots,r_N)^{\ST},~~ \bc=(c_1,c_2,\cdots,c_N)^{T}
\label{rc}
\end{equation}
\end{subequations}
and we suppose $\Ga$ is the canonical form of $\bk$.

\subsubsection{List of notations}\label{sec:3.2.2}

Based on Sec.\ref{sec:3.2.1} we only need to
consider three basic cases of $\Ga$:
\begin{subequations}\label{ga-cases}
\begin{align}
&\Ga^{\tyb{N}}_{\ty{D}}(\{k_j\}^{N}_{1})=\mathrm{Diag}(k_1, k_2, \cdots, k_N),~~(k_i^2\neq k_j^2, ~k_i\neq 0),\\
&\Ga^{\tyb{N}}_{\ty{J}}(k_1)
=\left(\begin{array}{cccccc}
k_1 & 0    & 0   & \cdots & 0   & 0 \\
1   & k_1  & 0   & \cdots & 0   & 0 \\
0   & 1  & k_1  & \cdots & 0   & 0 \\
\vdots &\vdots &\vdots &\vdots &\vdots &\vdots \\
0   & 0    & 0   & \cdots & 1   & k_1
\end{array}\right)=\bT^{\tyb{N}}[k_1],\\
 &\Ga^{\tyb{N}}_{\ty{G}}=\mathrm{Diag}\bigl(\Ga^{\tyb{N$_1$}}_{\ty{D}}(\{k_j\}^{N_1}_{1}),
\Ga^{\tyb{N$_2$}}_{\ty{J}}(k_{N_1+1}),\Ga^{\tyb{N$_3$}}_{\ty{J}}(k_{N_1+2}),\cdots,
\Ga^{\tyb{N$_s$}}_{\ty{J}}(k_{N_1+(s-1)})\bigr),
\label{Ga-g}
\end{align}
where $\sum^s_{j=1}N_j=N$.
\end{subequations}
The subscripts $_D$, $_J$ and $_G$ correspond to
the cases of $\Ga$ being diagonal, being of Jordan block and generic canonical form, respectively.
For convenience of the later discussions, let us collect some notations below.
\begin{subequations}\label{notations}
\begin{align}
& N\mathrm{\hbox{-}th~order~vector:}~~{\mbe^{\tyb{N}}}=(1,1,1, \cdots, 1)^{\ST},\\
& N\mathrm{\hbox{-}th~order~vector:}~~{\mathbf{e}_{1}^{\tyb{N}}}=(1,0,0, \cdots, 0)^{\ST},\label{IJ}\\
& N\mathrm{\hbox{-}th~order~vector:}~~g^{\tyb{N}}(a)=\Bigl(\frac{1}{a},\frac{-1}{a^2},\frac{1}{a^3}, \cdots,\frac{(-1)^{N-1}}{a^{N}}\Bigr)^{\ST},\\
& N\times N ~\mathrm{matrix:}~~\bG^{\tyb{N}}_{\ty{D}}(\{k_j\}^{N}_{1})
=(G_{i,j})_{N\times N},~~~G_{i,j}=\frac{1-g/{(K_iK_j)}}{k_i+k_j}.\\
& N\times N ~\mathrm{matrix:}~~\bH^{\tyb{N}}_{\ty{J}}(\{c_j\}^{N}_{1})
=\left(\begin{array}{ccccc}
c_1 & \cdots  & c_{N-2}  & c_{N-1} & c_N\\
c_2 & \cdots & c_{N-1}  & c_N & 0\\
c_3 &\cdots & c_N & 0 & 0\\
\vdots &\vdots & \vdots & \vdots & \vdots\\
c_N & \cdots & 0 & 0 & 0
\end{array}
\right).
\end{align}
\end{subequations}
$\bH^{\tyb{N}}_{\ty{J}}(\{c_j\}^{N}_{1})$ is called a skew LTT, which is also introduced in Appendix \ref{A:1}.

\subsubsection{Solutions to the matrix system \eqref{ell-relation-2-2} (i.e. \eqref{SE-cano-2})}
\label{sec:3.2.3}

We have discussed solutions to the matrix system \eqref{ell-relation-2-2} (i.e. \ref{SE-cano-2}) in Sec.\ref{sec:2.2}.
Here we list them using the notations given in Sec.\ref{sec:3.2.2}.

\begin{Proposition}\label{P:ellip}
The matrix system \eqref{ell-relation-2-2} admits the following three cases of solutions:\\
(1). Diagonal case:
\begin{equation}
\Ga=\Ga^{\tyb{N}}_{\ty{D}}(\{k_j\}^{N}_{1}),~~\bK=\Ga^{\tyb{N}}_{\ty{D}}(\{K_j\}^{N}_{1}),
\end{equation}
where
\begin{equation}\label{ell-relation-3}
  k_j^2=K_j+3e_1 +gK_j^{-1},~~j=1,2,\cdots,N.
\end{equation}
(2). Jordan block case:
\begin{equation}
\Ga=\Ga^{\tyb{N}}_{\ty{J}}(k_1),~~ \bK=\bT^{\tyb{N}}[K(k_1)].
\label{Ga-K-jordan}
\end{equation}
(3). Generic case:
\begin{subequations}
\label{Ga-gen-T}
\begin{align}
\Ga &=\Ga^{\tyb{N}}_{\ty{G}}, \\
\bK &= \mathrm{Diag}\bigl(\Ga^{\tyb{N$_1$}}_{\ty{D}}(\{K_j\}^{N_1}_{1}),
\bT^{\tyb{N$_2$}}[K(k_{N_1+1})],
\cdots,
\bT^{\tyb{N$_s$}}[K(k_{N_1+(s-1)}])\bigr).
\end{align}
\end{subequations}
\end{Proposition}

\subsubsection{Solutions to  \eqref{SE-cano-1}}

Now let us come to the solutions to the Sylvester equation \eqref{SE-cano-1}.

\vskip 5pt
\noindent
\textit{Case 1.~} $\Ga=\Ga^{\tyb{N}}_{\ty{D}}(\{k_j\}^{N}_{1})$.

Solution to \eqref{SE-cano-1} is given by
\begin{subequations}
\begin{align}
\bM &
=\bF \bG^{\tyb{N}}_{\ty{D}}(\{k_j\}^{N}_{1})\bH
=\Bigl(\frac{1-g/{(K_iK_j)}}{k_i+k_j}\,r_i c_j\Bigr)_{N\times N},
\end{align}
where
\begin{equation}
\bF=\mathrm{Diag}(r_1,r_2,\cdots, r_N),~~\bH=\mathrm{Diag}(c_1,c_2,\cdots, c_N).
\end{equation}
\label{sol-diag}
\end{subequations}

\vskip 5pt
\noindent \textit{Case 2.~}
$\Ga=\Ga^{\tyb{N}}_{\ty{J}}(k_1)$.

This is also referred to as the Jordan block case.
In this case, $\Ga$ and $\bK$ take \eqref{Ga-K-jordan}.
We suppose
\begin{equation}\label{factor-M}
\bM=\bF  \bG  \bH,
\end{equation}
where
\begin{equation}
\bF=\bT^{\tyb{N}}(\{r_j\}^{N}_{1}),~~\bH=\bH^{\tyb{N}}(\{c_j\}^{N}_{1}),
\end{equation}
and  $\bG$ is a $N\times N$ unknown matrix.
$\br$ and $\bc$ can be  expressed through $\bF$ and $\bH$ as
\begin{equation}
\br =\bF\,\mbe_1,~~~{\bc}=\bH\,\mbe_1,
\end{equation}
where $\mbe_1=\mbe_1^{\tyb{N}}$ is defined in \eqref{IJ}.
Then, the equation \eqref{SE-cano-1} is written as
\begin{equation}
\bF \bG\bH \Ga  +\Ga  \bF \bG\bH
=\bF\,\mbe_1\, \mbe_1^{\ST}\,
\bH-g\bK^{-1}\bF\,
\mbe_1\, \mbe_1^{\ST}\, \bH\bK^{-1}.
\end{equation}
Further, thanks to proposition \ref{prop-A-5}, one has
\begin{equation}
\bF \bG {\Ga}^{\ST} \bH + \bF \Ga \bG\bH
=\bF\, \mbe_1\, \mbe_1^{\ST}\,
\bH-g\bF\bK^{-1}\,\mbe_1\, \mbe_1^{\ST}\, {\bK^{-1}}^{\ST}\bH,
\end{equation}
and then
\begin{equation}
\bG {\Ga}^{\ST}  +  \Ga \bG =\mbe_1 \, \mbe_1^{\ST}-g\bK^{-1}\,\mbe_1 \, \mbe_1^{\ST}\, {\bK^{-1}}^{\ST}. \label{G-J-eq}
\end{equation}
To solve it, we set
\begin{equation}
\bG=(G_1,G_2,\cdots,G_N) \label{G-J-c-1}
\end{equation}
with column vectors $\{G_j\}$.
\eqref{G-J-eq} is then written into the following equation set,
\begin{subequations}
\begin{align}
& (k_1\bI+ \Ga) G_1= \mbe_1-\frac{g}{ K_1}A_1,\\
& (k_1\bI+ \Ga) G_{j+1}+G_{j}=-\frac{g}{j!} (\partial_{k_1}^j\frac{1}{K_1})A_1,~~ ~(j=1,2,\cdots, N-1),
\end{align}
where $K_1=K(k_1)$ and
\begin{equation}
A_1=(a_1,a_2,\cdots,a_N)^{\ST},~~a_j=\frac{1}{(j-1)!}\partial_{k_1}^{j-1}\frac{1}{K_1},~~(j=1,2,\cdots, N).
\end{equation}
\end{subequations}
The above is solved by
\begin{align}
 G_{j}& =\frac{\partial_a^{j-1} g^{\tyb{N}}(a)|_{a=2k_1}}{(j-1)!}+g\sum^{j}_{i=1}\frac{(-1)^i}{(j-i)!}
 (\partial_{k_1}^{j-i}\frac{1}{K_1}){\Ga^{\tyb{N}}_{\ty{J}}(2k_1)}^{-i}\, A_1,
~~(j=1,2,\cdots,N).
\label{g-C2-2}
\end{align}

\vskip 5pt
\noindent
\textit{Case 3.~} $\Ga=\Ga^{\tyb{N}}_{\ty{G}}$.

In this case, we
still suppose the factorization \eqref{factor-M}, where
\begin{subequations}
\begin{align}
&\bF=\mathrm{Diag}\bigl(
\Ga^{\tyb{N$_1$}}_{\ty{D}}(\{k_j\}^{N_1}_{1}),
\bT^{\tyb{N$_2$}}(\{r_j\}^{N_1+N_2}_{N_1+1}),
\cdots,
\bT^{\tyb{N$_s$}}(\{r_j\}^{N_1+N_2+\cdots+N_s}_{N_1+N_2+\cdots+N_{s-1}+1})
\bigr),\\
&\bH=\mathrm{Diag}\bigl(
\bH^{\tyb{N$_1$}}_{\ty{D}}(\{c_j\}^{N_1}_{1}),
\bH^{\tyb{N$_2$}}_{\ty{J}}(\{c_j\}^{N_1+N_2}_{N_1+1}),
\cdots,
\bH^{\tyb{N$_s$}}_{\ty{J}}(\{c_j\}^{N_1+N_2+\cdots+N_s}_{N_1+N_2+\cdots+N_{s-1}+1})\bigr),
\end{align}
$\bG$ is a symmetric matrix with block structure
\begin{equation}
\bG=\bG^{\ST}=(\bG_{i,j})_{s\times s}
\end{equation}
\end{subequations}
and each $\bG_{i,j}$ is a $N_i\times N_j$ matrix.
It is not difficult to find $\bG_{i,j}(=\bG_{j,i}^T)$ are given by
\begin{subequations}\label{gene-sol}
\begin{align}
&\bG_{1,1}=\bG^{\tyb{N}}_{\ty{D}}(\{k_j\}^{N}_{1}),\label{gene-sol-1}\\
& \bG_{i,j}=(G_{i1},G_{i2},\cdots,G_{iN_j}),~~(1<i\leq j\leq s),
\end{align}
\end{subequations}
with
\begin{align*}
&G_{1l}=(-1)^{l-1}\Bigl[\bI
-g\sum^{l}_{i=1}\frac{(-1)^{l-i}}{(l-i)!}\Bigl(\partial_{k}^{l-i}\frac{1}{K(k)}\Bigr)\Bigr|_{k=k_{N_1+1}}
(\Ga_{\ty{D}}^{\tyb{N$_1$}}(\{K_j\}^{N_1}_{1}))^{-1} \Bigr]
\Ga_{\ty{D}}^{\tyb{N$_1$}}(\{\alpha_j\})\mbe^{\tyb{N$_1$}},\\
&G_{il}=\frac{\partial_{\beta_{ij}}^{l-1} g^{\tyb{N$_i$}}(\beta_{ij})}{(l-1)!}\\
&~~~~~~~+g\sum^{l}_{m=1}\frac{(-1)^m}{(l-m)!}\Bigl(\partial_{k}^{l-m}\frac{1}{K(k)}\Bigr)(\Ga_{\ty{J}}^{\tyb{N$_i$}}(\beta_{ij}))^{-m}
\bT^{\tyb{N$_i$}}(\frac{1}{K(k)}) \mbe_1^{\tyb{N$_i$}}\Bigr|_{k=k_{N_1+(j-1)}}
\end{align*}
for $1<i\leq j\leq s$, where
$\alpha_j=\frac{1}{k_j +k_{N_1+(j-1)}}$, $\beta_{ij}=k_{N_1+(i-1)}+k_{N_1+(j-1)}$
and $k$ and $K(k)$ satisfy the elliptic curve \eqref{ec-sca}, i.e.,
$k^2=K^2(k)+3e_1+g/K(k)$.

\subsection{Infinite matrix $\bS$}
\subsubsection{Recurrence relation of $\S{i,j}$}

Let us go back to the Sylvester equation \eqref{SE-3a} and the matrix relation \eqref{ec-mat}.
Using the elements $\{\bM,\bk,\bK,\br,\bc\}$ in  \eqref{SE-3a} and  \eqref{ec-mat},
we introduce an $\infty\times\infty$ matrix $\bS=(S^{(i,j)})_{\infty\times\infty},~i,j\in \mathbb{Z},$
where the elements $S^{(i,j)}$ are defined as(cf. \cite{NP-JNMP-2003})
\begin{subequations}\label{Sij}
\begin{eqnarray}
&&S^{(2i,2j)}= {\bc}^{\ST} \,\bK^j(\bI+ \bM)^{-1} \bK^i \br,\label{Sij-a}\\
&& S^{(2i+1,2j)}= {\bc}^{\ST} \,\bK^j(\bI+ \bM)^{-1} \bk \bK^i \br,\label{Sij-b}\\
&&S^{(2i,2j+1)}= {\bc}^{\ST} \,\bK^j\bk (\bI+ \bM)^{-1} \bK^i \br,\label{Sij-c}\\
&& S^{(2i+1,2j+1)}= {\bc}^{\ST} \,\bK^j\bk (\bI+ \bM)^{-1} \bk \bK^i
\br.\label{Sij-d}
\end{eqnarray}
\end{subequations}

For these elements we have the following.

\begin{Proposition}\label{prop-2}
For the saclar functions $S^{(i,j)}$ defined in \eqref{Sij} with $(\bM, \bK, \bk, \br,\bc)$
satisfying the Sylvester equation \eqref{SE-3a} and the matrix relation \eqref{ec-mat},
we have the following relations,
\begin{subequations}
\begin{eqnarray}
&&\S{i,j+2s}=\S{i+2s,j}-\sum^{s-1}_{l=0}(\S{2s-2l-1,j}\S{i,2l}- \S{2s-2l-2,j}\S{i,2l+1}),\label{Sij-k+}\\
&&\S{i,j-2s}=\S{i-2s,j}+\sum^{s}_{l=1}(\S{-2s+2l-1,j}\S{i,-2l}- \S{2l-2s-2,j}\S{i,-2l+1}),\label{Sij-k-}
\end{eqnarray}
\end{subequations}
where $s=1,2,\cdots$. In particular, when $s=1,$ one has
\begin{subequations}\label{Sij-k1}
\begin{eqnarray}
&&\S{i,j+2}=\S{i+2,j}-\S{1,j}\S{i,0}+\S{0,j}\S{i,1},\label{Sij-k+1}\\
&&\S{i,j-2}=\S{i-2,j}+\S{-1,j}\S{i,-2}-\S{-2,j}\S{i,-1}.\label{Sij-k-1}
\end{eqnarray}
\end{subequations}
\end{Proposition}

\begin{proof}
First, from the Sylvester equation \eqref{SE-3a} we have the following relation
\begin{equation}
\bk^s \bM -(-1)^{s} \bM \bk^s=\sum^{s-1}_{j=0}(-1)^j\bk^{s-j-1}( \br
{\bc}^{\ST}-g\bK^{-1}\br {\bc}^{\ST} \bK^{-1}) \bk^j,~~~(s=1,2,\cdots).
\label{Mk-rec}
\end{equation}
In fact, the Sylvester equation \eqref{SE-3a} itself is the case of  $s=1$ of \eqref{Mk-rec}, while
\eqref{s-2} is the case of $s=2$.
Making use of  mathematical inductive approach we can reach \eqref{Mk-rec}.
Similarly, from \eqref{SE-3b} one has a parallel result
\begin{subequations}
\begin{equation}
\bK^s \bM - \bM \bK^s=\sum^{s-1}_{j=0}\bK^{s-j-1}( \bk\br
{\bc}^{\ST}-\br {\bc}^{\ST} \bk) \bK^j, \label{MK-rec}
\end{equation}
or
\begin{equation}
\bM\bK^{-s} - \bK^{-s}\bM=\sum^{s}_{j=1}\bK^{-(s-j+1)}( \bk\br
{\bc}^{\ST}-\br {\bc}^{\ST} \bk) \bK^{-j},~~~(s=1,2,\cdots). \label{MK-rec--}
\end{equation}
\end{subequations}
Now let us prove the relation \eqref{Sij-k+}. We introduce the auxiliary vectors
\begin{equation} \label{ui}
\bu^{(2i)}=(\bI+\bM)^{-1} \bK^i \br,~~\bu^{(2i+1)}=(\bI+\bM)^{-1}\bk
\bK^i \br,~~ i\in \mathbb{Z}.
\end{equation}
From this we immediately have
\begin{subequations}
\begin{eqnarray}
&&\bK^s \bu^{(2i)} + \bK^s \bM \bu^{(2i)}=\bK^{s+i} \br,\\
&&\bK^s \bu^{(2i+1)} + \bK^s \bM \bu^{(2i+1)}=\bk \bK^{s+i} \br.
\end{eqnarray}
\end{subequations}
Replacing $\bK^s \bM$ using the relation \eqref{MK-rec}, one finds
\begin{subequations}\label{re-u+}
\begin{eqnarray}
&&(\bI+\bM)\bK^{s} \bu^{(2i)}=\bK^{s+i} \br
-\sum^{s-1}_{l=0} \bK^{s-l-1} (\bk\br {\bc}^{\ST}-\br {\bc}^{\ST} \bk) \bK^l \bu^{(2i)},\label{re-u+1}\\
&&(\bI+\bM)\bK^{s} \bu^{(2i+1)}=\bk \bK^{s+i} \br -\sum^{s-1}_{l=0}
\bK^{s-l-1} (\bk\br {\bc}^{\ST}-\br {\bc}^{\ST} \bk) \bK^l \bu^{(2i+1)}.\label{re-u+2}
\end{eqnarray}
\end{subequations}
These relations, left-multiplied  by ${\bc}^{\ST}\bK^j(\bI+\bM)^{-1}$, yield
\begin{align*}
&\S{2i,2j+2s}=\S{2i+2s,2j}-\sum^{s-1}_{l=0}(\S{2s-2l- 1,2j}\S{2i,2l}-\S{2s-2l-2,2j}\S{2i,2l+1}), \\
&\S{2i+1,2j+2s}=\S{2i+2s+1,2j}-\sum^{s-1}_{l=0}(\S{2s-2l-1,2j}\S{2i+1,2l}-\S{2s-2l-2,2j}\S{2i+1,2l+1}),
\end{align*}
and
 left-multiplied  by  ${\bc}^{\ST}\bK^j\bk (\bI+\bM)^{-1}$,  yield
\begin{align*}
&\S{2i,2j+2s+1}=\S{2i+2s,2j+1}-\sum^{s-1}_{l=0}(\S{2s-2l-1,2j+1}\S{2i,2l}-\S{2s-2l-2,2j+1}\S{2i,2l+1}),\nonumber\\
&\S{2i+1,2j+2s+1}=\S{2i+2s+1,2j+1}-\sum^{s-1}_{l=0}(\S{2s-2l-1,2j+1}\S{2i+1,2l}-\S{2s-2l-2,2j+1}\S{2i+1,2l+1}).
\end{align*}
They are merged to the relation \eqref{Sij-k+}.

The relation \eqref{Sij-k-} can be proved in a similar procedure,
in which we use the following counterpart of \eqref{re-u+}:
\begin{subequations}\label{re-u-}
\begin{eqnarray}
&&(\bI+\bM)\bK^{-s} \bu^{(2i)}=\bK^{-s+i} \br
+\sum^{s}_{l=1} \bK^{-(s-l+1)} (\bk\br {\bc}^{\ST}-\br {\bc}^{\ST} \bk) \bK^{-l} \bu^{(2i)},\\
&&(\bI+\bM)\bK^{-s} \bu^{(2i+1)}=\bk \bK^{-s+i} \br+\sum^{s}_{l=1}
\bK^{-(s-l+1)} (\bk\br {\bc}^{\ST}-\br {\bc}^{\ST} \bk) \bK^{-l} \bu^{(2i+1)},
\end{eqnarray}
\end{subequations}
with $s=1,2,\cdots$.
We note that \eqref{Sij-k1} are corresponding to the algebraic relations (2.13) in Ref.\cite{NP-JNMP-2003}.
\end{proof}

\subsubsection{Invariance and symmetry property of $\S{i,j}$ }\label{sec:2.3.2}

The matrix $\bS$ (or the element $\S{i,j}$) keeps invariant under the similarity transformation \eqref{trans-sim}.

We have shown that the matrix relation \eqref{ell-relation-2} and the Sylvester equation \eqref{SE-3a}
keep invariant formally in terms of the similarity transformation \eqref{trans-sim} and notations \eqref{Mrs-1}.
Using \eqref{trans-sim} and \eqref{Mrs-1} one can rewrite  \eqref{Sij} and find
\begin{subequations}\label{Sij-sim}
\begin{eqnarray}
&&S^{(2i,2j)}= {\bc}^{\ST}_1 \,\bK_1^j(\bI+ \bM_1)^{-1} \bK_1^i \br_1,\label{Sij-sim-a}\\
&& S^{(2i+1,2j)}= {\bc}^{\ST}_1 \,\bK_1^j(\bI+ \bM_1)^{-1} \bk_1 \bK_1^i \br_1,\label{Sij-sim-b}\\
&&S^{(2i,2j+1)}= {\bc}^{\ST}_1 \,\bK_1^j\bk_1 (\bI+ \bM_1)^{-1} \bK_1^i \br_1,\label{Sij-sim-c}\\
&& S^{(2i+1,2j+1)}= {\bc}^{\ST} _1\,\bK_1^j\bk_1 (\bI+ \bM_1)^{-1} \bk_1 \bK_1^i \br_1.\label{Sij-sim-d}
\end{eqnarray}
\end{subequations}

In addition, we have the following symmetry property.

\begin{Proposition}\label{prop-sym}
Suppose that $ \bM, \bK,\bk, \br, \bs$ satisfy the Sylvester equation \eqref{SE-3a} together with the matrix system \eqref{ec-mat} in which
$\mathcal{E}(\bk)\cap \mathcal{E}(-\bk)=\varnothing$ and $\mathcal{E}(g\bK^{-1})\cap \mathcal{E}(\bK)=\varnothing$.
Then the
scalar elements $\S{i,j}$ defined by \eqref{Sij} satisfy the
symmetry property
\begin{equation}\label{Sij=Sji}
\S{i,j}=\S{j,i},
\end{equation}
i.e. the infinite matrix $\bS$ is symmetric.
\end{Proposition}

This can be proved using the property given in proposition \ref{prop-A-5} and following the procedure
described in Appendix C in Ref.\cite{XZZ-2014-JNMP}. Here we skip the proof.

Hereafter we always require that  $\{ \bM, \bK,\bk, \br, \bs\}$ satisfy the
assumption of proposition \ref{prop-sym}, under which we
proceed further discussions.

\section{The elliptic lattice potential KdV system}\label{sec-4}

\subsection{Discrete dispersion relation and recurrence relations}\label{sec:4.1}

Now let us insert discrete dispersion relation on $\br$ as follows
\begin{equation}
(a \bI- \bk)\wt{\br}  =(a \bI+\bk) \br,~~ (b \bI-\bk)\wh{\br}  =(b \bI+\bk) \br,~~ a,b \notin \mathcal{E}(\pm\bk),
\label{dr}
\end{equation}
while we take $\bc$ to be a constant vector.

Employing a similar procedure as done in \cite{ZZ-SAM-2013}, from the
Sylvester equation \eqref{SE-3a}, matrix system \eqref{ec-mat} and the dispersion relation
\eqref{dr}, one can derive the following shift relation of $\bM$,
\begin{subequations}\label{cond-qM-th}
\begin{align}
& (a \bI-\bk) \wt{\bM}=(a \bI+\bk)\bM, \label{cond-qM-t}\\
& (b \bI- \bk) \wh{\bM}=(b \bI+\bk)  \bM, \label{cond-qM-h}
\end{align}
\end{subequations}
and
\begin{subequations}
\begin{align}
& \wt{\bM}(a\bI+\bk)-(a\bI+\bk)\bM=\wt{\br}{\bc}^{\ST}-g\bK^{-1}\wt{\br} {\bc}^{\ST} \bK^{-1}, \label{eq:qM-dyna-1}\\
& (a\bI-\bk)\wt{\bM}-\bM(a\bI-\bk)=\br{\bc}^{\ST}-g\bK^{-1}\br{\bc}^{\ST} \bK^{-1}, \label{eq:qM-dyna-2}\\
& \wh{\bM}(b\bI+\bk)-(b\bI+\bk)\bM=\wh{\br}{\bc}^{\ST}-g\bK^{-1}\wh{\br} {\bc}^{\ST} \bK^{-1},\label{eq:qM-dyna-3}\\
& (b\bI-\bk)\wh{\bM}-\bM(b\bI-\bk)=\br{\bc}^{\ST}-g\bK^{-1}\br {\bc}^{\ST}
\bK^{-1}.\label{eq:qM-dyna-4}
\end{align}
\label{qM-dyna}
\end{subequations}

These relations lead to the following results.
\begin{Proposition}
\label{P:evo-sij}
Under the assumption of proposition \ref{prop-sym} and the dispersion relation \eqref{dr},
the scalar functions $\S{i,j}$ defined by \eqref{Sij} enjoy the following recurrence relations:
\begin{subequations}\label{Sij-shift}
\begin{eqnarray}
&& a \wt{S}^{(2i,2j)}-\wt{S}^{(2i,2j+1)}=a S^{(2i,2j)}+S^{(2i+1,2j)}-\wt{S}^{(2i,0)}S^{(0,2j)}+g\wt{S}^{(2i,-2)}{S}^{(-2,2j)},\label{Sij-shift-a}\\
&&  a {S}^{(2i,2j)}+{S}^{(2i,2j+1)}=a \wt{S}^{(2i,2j)}-\wt{S}^{(2i+1,2j)}+{S}^{(2i,0)}\wt S^{(0,2j)}-g{S}^{(2i,-2)}\wt{S}^{(-2,2j)},\label{Sij-shift-b}\\
&& b \wh{S}^{(2i,2j)}-\wh{S}^{(2i,2j+1)}=b S^{(2i,2j)}+S^{(2i+1,2j)}-\wh{S}^{(2i,0)}S^{(0,2j)}+g\wh{S}^{(2i,-2)}{S}^{(-2,2j)},\label{Sij-shift-e}\\
&&  b {S}^{(2i,2j)}+{S}^{(2i,2j+1)}=b \wh{S}^{(2i,2j)}-\wh{S}^{(2i+1,2j)}+{S}^{(2i,0)}\wh S^{(0,2j)}-g{S}^{(2i,-2)}\wh{S}^{(-2,2j)},\label{Sij-shift-f}\\
&& a \wt{S}^{(2i+1,2j)}-\wt{S}^{(2i+1,2j+1)}=a {S}^{(2i+1,2j)}+{S}^{(2i+2,2j)}-\wt{S}^{(2i+1,0)}S^{(0,2j)}\nonumber\\
&&~~~~~~~~~~~~~~~~~~~~~~~~~~~~~~~~+g{S}^{(2i-2,2j)}+g\wt{S}^{(2i+1,-2)}S^{(-2,2j)}+3e_1{S}^{(2i,2j)},\label{Sij-shift-c}\\
&& a {S}^{(2i+1,2j)}+{S}^{(2i+1,2j+1)}=a \wt{S}^{(2i+1,2j)}-\wt{S}^{(2i+2,2j)}+{S}^{(2i+1,0)}\wt S^{(0,2j)}\nonumber\\
&&~~~~~~~~~~~~~~~~~~~~~~~~~~~~~~~~-g\wt{S}^{(2i-2,2j)}-g\wt{S}^{(2i+1,-2)}S^{(-2,2j)}-3e_1\wt{S}^{(2i,2j)},\label{Sij-shift-d}\\
&& b \wh{S}^{(2i+1,2j)}-\wh{S}^{(2i+1,2j+1)}=b {S}^{(2i+1,2j)}+{S}^{(2i+2,2j)}-\wh{S}^{(2i+1,0)}S^{(0,2j)}\nonumber\\
&&~~~~~~~~~~~~~~~~~~~~~~~~~~~~~~~~+g{S}^{(2i-2,2j)}+g\wh{S}^{(2i+1,-2)}S^{(-2,2j)}+3e_1{S}^{(2i,2j)},\label{Sij-shift-g}\\
&& b {S}^{(2i+1,2j)}+{S}^{(2i+1,2j+1)}=b \wh{S}^{(2i+1,2j)}-\wh{S}^{(2i+2,2j)}+{S}^{(2i+1,0)}\wh S^{(0,2j)}\nonumber\\
&&~~~~~~~~~~~~~~~~~~~~~~~~~~~~~~~~-g\wh{S}^{(2i-2,2j)}-g\wh{S}^{(2i+1,-2)}S^{(-2,2j)}-3e_1\wh{S}^{(2i,2j)}.\label{Sij-shift-h}
\end{eqnarray}
\end{subequations}
\end{Proposition}

The proof of this proposition is similar to the one for theorem 2 in Ref.\cite{ZZ-SAM-2013}.
Here we skip the details. We also note that these relations are corresponding  to the discrete matrix Riccati type of relations (2.12)
in Ref.\cite{NP-JNMP-2003}.

\subsection{Elliptic lattice equations}\label{sec:4.2}

To obtain  elliptic lattice equations, we introduce scalar functions (cf.\cite{NP-JNMP-2003})
\begin{equation}\label{uvswh}
   u=S^{(0,0)},~~s=S^{(-2,0)},~~h=S^{(-2,-2)},~~v=1-S^{(-1,0)},~~w=1+S^{(-2,1)}.
\end{equation}
It then follows from \eqref{Sij-shift} that
\begin{subequations}\label{re-la}
\begin{align}
a(u-\wt u)&=-(\wt{S}^{(0,1)}+S^{(0,1)})-gs\wt s+u\wt u,\label{re-la-1}\\
a(h-\wt h)&=-(\wt{S}^{(-1,-2)}+S^{(-1,-2)})-gh\wt h+s\wt s, \label{re-la-2}\\
a(s-\wt{s})&=g\wt{h}s+\wt{w}-v-u\wt{s},\label{re-la-41} \\
a(s-\wt{s})&=g{h}\wt s+{w}-\wt v-\wt u{s}, \label{re-la-4}\\
a(v-\wt v)&=\wt{S}^{(-1,1)}+3e_1 s+u\wt v+gvh+gs({S}^{(-2,-1)}+\wt{S}^{(-2,-1)}),\label{re-la-42}\\
a(v-\wt v)&={S}^{(-1,1)}+3e_1 \wt s+\wt u v+g\wt v\wt h+g\wt s({S}^{(-2,-1)}+\wt{S}^{(-2,-1)}),\label{re-la-5}\\
a(w-\wt w)&=\wt{S}^{(-1,1)}+3e_1 s+uw+gh \wt w-s({S}^{(0,1)}+\wt{S}^{(0,1)}),\label{re-la-43}\\
a(w-\wt w)&={S}^{(-1,1)}+3e_1 \wt s+\wt u\wt w+g\wt h w-\wt s({S}^{(0,1)}+{S}^{(0,1)}),\label{re-la-6}
\end{align}
where we have made use of the symmetric property $\S{i,j}=\S{j,i}$ and the relation \eqref{Sij-k-1}
with $(i,j)=(0,-2)$,
i.e.
\[
\S{0,-4}=\S{-2,-2}+\S{-1,-2}\S{0,-2}-\S{-2,-2}\S{0,-1}=hv+s\S{-2,-1}.
\]
One more relation coming from \eqref{Sij-k-1} is
\begin{align}
s{S}^{(-1,1)}&=1-vw,\label{re-la-7}
\end{align}
when taking $(i,j)=(0,1)$.
\end{subequations}
The shift relations w.r.t. $\{b,\wh{\cdot}~\}$ can be obtained by
interchanging the relations $\{a,\wt{\cdot}~\}$ with $\{b,\wh{\cdot}~\}$ in \eqref{re-la}.

With these relations in hand, we can combine them, eliminate $\S{0,1}, \S{-1,-2}, h$ and $v$ and reach
the closed form of elliptic lattice equations of the variables $u,s$ and $w$, (cf. \cite{NP-JNMP-2003})
\begin{subequations}\label{ell-la}
\begin{align}
&(a+b+u-\wh{\wt{u}})(a-b+\wh{u}-\wt{u})=a^2-b^2+g(\wt s-\wh s)(\wh{\wt{s}}-s),\label{ell-la-1}\\
&(\wh{\wt{s}}-s)(\wt{w}-\wh{w})=[(a+u)\wt s-(b+u)\wh s]\wh{\wt{s}}-[(a-\wh{\wt{u}})\wh s-(b-\wh{\wt{u}})\wt s]s,\label{ell-la-2}\\
&(\wh{s}-\wt s)(\wh{\wt{w}}-{w})=[(a-\wt u) s+(b+\wt u)\wh{\wt{s}}]\wh{s}-[(a+\wh{u})\wh{\wt{s}}+(b-\wh{u}) s]\wt s, \label{ell-la-3}\\
&(a+ u-\frac{\wt w}{\wt s})({a-\wt u+\frac{w}{s}})=a^2-R(s\t s), \label{ell-la-4}\\
&(b+ u-\frac{\wh w}{\wh s})({b-\wh u+\frac{w}{s}})=b^2-R(s\h s),\label{ell-la-5}
\end{align}
\end{subequations}
where
\begin{equation}\label{ell-curve-2-4}
    y^2=R(x)=\frac{1}{x}+3e+gx.
\end{equation}
This is an elliptic generalization of the lpKdV equation,
namely that if $g=0$ the first equation \eqref{ell-la-1} decouples and becomes the standard lattice potential KdV equation.

Furthermore, solutions to the elpKdV system \eqref{ell-la} can be expressed in the explicit structure:
\begin{subequations} \label{ell-lattice}
\begin{align}
&u=S^{(0,0)}={\bc}^{\ST} \,(\bI+ \bM)^{-1}  \br,\\
&s=S^{(-2,0)}={\bc}^{\ST} \,(\bI+ \bM)^{-1} \bK^{-1} \br,\\
&h=S^{(-2,-2)}={\bc}^{\ST} \,\bK^{-1}(\bI+ \bM)^{-1} \bK^{-1} \br,\\
&v=1-S^{(-1,0)}=1-{\bc}^{\ST} \,(\bI+ \bM)^{-1} \bk \bK^{-1} \br,\\
&w=1+S^{(-2,1)}={\bc}^{\ST} \,\bk (\bI+ \bM)\bK^{-1} \br.
\end{align}
\end{subequations}
We only need to solve out $\br$ from \eqref{dr} with $\bk$ taking three cases in \eqref{ga-cases}. We will list out explicit forms of
$\br$ in Appendix \ref{A:2}.

\subsection{Lax pair for the elpKdV system \eqref{ell-lattice}}

Rewriting \eqref{ui} as
\begin{subequations} \label{ui-2-}
\begin{eqnarray}
&&\bK^i \br=(\bI+\bM)\bu^{(2i)} ,\label{ui-2-1}\\
&&\bk \bK^i \br=(\bI+\bM)\bu^{(2i+1)},~~ i\in
\mathbb{Z},\label{ui-2-2}
\end{eqnarray}
\end{subequations}
we can perform a tilde shift to \eqref{ui-2-1} and multiply the result by $(a\bI-\bk)$ to give:
\begin{equation}
\bK^i (a\bI+\bk)
\br=(\bI+\bM)(a\bI-\bk)\wt{\bu}^{(2i)}+(\br{\bc}^{\ST}-g\bK^{-1}\br
{\bc}^{\ST} \bK^{-1})\wt{\bu}^{(2i)}, \label{ui-2-1-t}
\end{equation}
where we have also make use of the fact $\bk \bK=\bK\bk$, the dispersion relation \eqref{dr},
shift relation \eqref{cond-qM-t} and the Sylvester equation \eqref{SE-3a}.
This further brings
\begin{equation}\label{ui-2-1-2}
(a\bI-\bk)\wt{\bu}^{(2i)}=a{\bu}^{(2i)}+{\bu}^{(2i+1)}-{\t
S}^{(2i,0)}{\bu}^{(0)}+g{\t S}^{(2i,-2)}{\bu}^{(-2)}.
\end{equation}
Taking $i=0$ it reads
\begin{equation}\label{ui-2-1-3}
(a\bI-\bk)\wt{\bu}^{(0)}=(a-{\t
u}){\bu}^{(0)}+{\bu}^{(1)}+g{\t s}\bK^{-1}(w{\bu}^{(0)}-s{\bu}^{(1)}),
\end{equation}
where we have replaced ${\bu}^{(-2)}$ by
\begin{equation}\label{u-2}
   {\bu}^{(-2)}=\bK^{-1}(w{\bu}^{(0)}-s{\bu}^{(1)})
\end{equation}
coming from the relation \eqref{re-u+1} with $s=1,~i=-1$.

In a similar way, from \eqref{ui-2-2} we have
\begin{equation}\label{ui-2-2-3}
   (a\bI-\bk)\wt{\bu}^{(1)}=(\bK+g\wt w w\bK^{-1}) {\bu}^{(0)}+(3e_1+g\wt s s+a(u-\wt u)-u \wt u){\bu}^{(0)}+(a+u-g\wt w s\bK^{-1}){\bu}^{(1)},
\end{equation}
where
\begin{equation}\label{u+2}
   {\bu}^{(2)}=\bK{\bu}^{(0)}-\bS^{(0,1)}{\bu}^{(0)}+u{\bu}^{(1)}
\end{equation}
is used.

It is easy to check that under the similarity transformation \eqref{trans-sim}, equations
\eqref{ui-2-1-3} and \eqref{ui-2-2-3}
are formally invariant if we define ${\bu}^{(0)}_1=\bT{\bu}^{(0)},~{\bu}^{(1)}_1=\bT{\bu}^{(1)}$.
That means we can directly consider that $\bk$ is in its canonical form and $\bK$ is consequently defined by \eqref{ec-mat}.
Thus, the first row of $\bk$ will be $(k,0,0,\cdots,0)$ and the first rows of $\bK$ and $\bK^{-1}$ have to be
$(K,0,0,\cdots,0)$  and  $(1/K,0,0,\cdots,0)$ respectively, where $(k,K)$ obey the elliptic curve \eqref{ec-sca}.

Let us denote the first element of ${\bu}^{(0)}$ by $({\bu}^{(0)})_1$ and
the first element of ${\bu}^{(1)}$ by $({\bu}^{(1)})_1$,
and introduce the vector
 \begin{equation}\label{phi}
    \phi=\left(
           \begin{array}{c}
            ( {\bu}^{(0)})_1\\
            ( {\bu}^{(1)})_1\\
           \end{array}
         \right).
 \end{equation}
Then, from  \eqref{ui-2-1-3} and \eqref{ui-2-2-3} and their $(b,\wh{~~} )$ version
we obtain the following discrete linear system:
 \begin{subequations}\label{dlax pair}
\begin{align}
&(a-k)\wt\phi=L(K)\phi,\\
&(b-k)\wh\phi=M(K)\phi,
\end{align}
\end{subequations}
 where
 \begin{equation}\label{L(K)}
   L(K)=\left(
          \begin{array}{cc}
            a-\t u+\frac{g}{K}\t sw & 1- \frac{g}{K}\t ss\\
            K+3e+g\t ss+a(u-\t u)-\t uu+\frac{g}{K}\t ww & a+u-\frac{g}{K}\t ws \\
          \end{array}
        \right)
 \end{equation}
which is as same as in \cite{NP-JNMP-2003}, and $M(K)$ is the $(b,\wh{~~} )$ counterpart of $L(K)$.
The point $(k,K)$ obeys the elliptic curve \eqref{ec-sca} and here they play the roles of spectral parameters.
The compatibility condition
\begin{equation}\label{compat}
   \wh LM=\wt ML
\end{equation}
yields the whole ellpKdV system \eqref{ell-la} with the exception of eq.\eqref{ell-la-3}.
Here we nota that there can be mis-matched between a discrete system and the system derived from its Lax pair, for example,
the discrete Boussinesq-type equations \cite{ZZN-SAM-2012}.

\section{The elliptic potential KdV system}\label{sec-5}

In this section we try to derive a continuous elliptic potential KdV (eqKdV) system.
The procedure will be similar to the one for the KdV system in \cite{XZZ-2014-JNMP},
which can be viewed as a continuous version of Cauchy matrix approach,
where the recurrence relation of $\S{i,j}$ (see (2.3) in \cite{XZZ-2014-JNMP} and \eqref{Sij-k1} in this paper) will play key roles.

\subsection{Evolution of $\bM$}\label{sec:5.1}

We suppose that $\bM,\br, \bc$ are functions of $(x,t)$ while
$\bk$ is still a non-trivial constant matrix.
The dispersion relation is now defined through the evolution of $\br$
and $\bc$ as follows,
\begin{subequations}
\begin{align}
& \br_x=\bk \br,~~\bc_x=\bk^{\ST} \bc, \label{evo-rs-x} \\
& \br_t=4\bk^3 \br,~~\bc_t=4(\bk^{\ST})^3\bc. \label{evo-rs-t}
\end{align}
\label{evo-rs}
\end{subequations}
Taking the derivative of the Sylvester equation \eqref{SE-3a} w.r.t. $x$
and making using of \eqref{evo-rs-x} we have
\begin{align*}
\bk \bM_x+ \bM_x\bk &=\br_x \,{\bc}^{\ST}+ \br {\bc}^{\ST}_x-g\bK^{-1}\br_x {\bc}^{\ST} \bK^{-1}-g\bK^{-1}\br {\bc}^{\ST}_x \bK^{-1} \\
&=\bk\br{\bc}^{\ST}+\br{\bc}^{\ST}\bk-g\bK^{-1}\bk \br {\bc}^{\ST}
\bK^{-1}-g\bK^{-1}\br {\bc}^{\ST} \bk \bK^{-1},
\end{align*}
i.e.
\[\bk (\bM_x-\br{\bc}^{\ST}+g\bK^{-1}\br {\bc}^{\ST} \bK^{-1})+ (\bM_x-\br{\bc}^{\ST}+g\bK^{-1}\br {\bc}^{\ST} \bK^{-1})\bk=0, \]
where we have made use of the relation $\bk\bK=\bK\bk.$
Using proposition \ref{prop-1} this yields
\begin{equation}
\bM_x=\br{\bc}^{\ST}-g\bK^{-1}\br {\bc}^{\ST} \bK^{-1}, \label{evo-Mx}
\end{equation}
i.e.
\begin{equation}
\bM_x=\bk \bM+ \bM\bk, \label{Mx-1}
\end{equation}
if we use the Sylvester equation \eqref{SE-3a}.
In a similar way, for the time evolution of $\bM$, we have
\begin{equation}
\bM_t=4(\bk^3 \bM+ \bM\bk^3).
\label{Mt3}
\end{equation}

\subsection{Evolution of $\S{i,j}$}

With the evolution formulas \eqref{evo-rs}-\eqref{Mt3},
we can derive the evolution of $\S{i,j}$.
We make use of the auxiliary vectors $\bu^{(i)}$ defined in \eqref{ui}, i.e.
\begin{subequations}
\begin{eqnarray}
&&(\bI+\bM)\bu^{(2i)}= \bK^i \br,\label{ui-1}\\
&&(\bI+\bM)\bu^{(2i+1)}=\bk \bK^i \br.
\label{ui-2}
\end{eqnarray}
\end{subequations}
By them $\S{i,j}$ are expressed as
\begin{subequations}\label{Sij-3}
\begin{eqnarray}
&&S^{(2i,2j)}= {\bc}^{\ST} \,\bK^j\bu^{(2i)},\label{Sij-3-a}\\
&& S^{(2i+1,2j)}= {\bc}^{\ST} \,\bK^j\bu^{(2i+1)},\label{Sij-3-b}\\
&&S^{(2i,2j+1)}= {\bc}^{\ST} \,\bK^j\bk \bu^{(2i)},\label{Sij-3-c}\\
&& S^{(2i+1,2j+1)}= {\bc}^{\ST} \,\bK^j\bk \bu^{(2i+1)}.\label{Sij-3-d}
\end{eqnarray}
\end{subequations}
Taking $x$-derivative on \eqref{ui-1}  we have
\begin{equation}
    \bM_x\bu^{(2i)}+(\bI+\bM)\bu^{(2i)}_x=\bK^{i}\br_x=\bK^{i}\bk\br,
\end{equation}
and further, by substitution of \eqref{evo-Mx}, we have
\begin{equation}
    (\bI+\bM)\bu^{(2i)}_x=\bK^{i}\bk\br-(\br {\bc}^{\ST}-g\bK^{-1}\br {\bc}^{\ST} \bK^{-1} )\bu^{(2i)},
\end{equation}
which indicates the evolution of $\bu^{(2i)}$ in $x$-direction:
\begin{subequations}\label{evo-uix}
\begin{equation}
\bu^{(2i)}_x=\bu^{(2i+1)}-S^{(2i,0)}\bu^{(0)}+gS^{(2i,-2)}\bu^{(-2)}.
\end{equation}
Similarly we can derive out the evolution of $\bu^{(2i+1)}$ in $x$-direction
\begin{eqnarray}
\bu^{(2i+1)}_x=\bu^{(2i+2)}-S^{(2i+1,0)}\bu^{(0)}+gS^{(2i+1,-2)}\bu^{(-2)}+3e_1\bu^{(2i)}+g\bu^{(2i-2)},
\end{eqnarray}
\end{subequations}
and the evolution of $\bu^{(i)}$ in $t$-direction:
\begin{subequations}\label{evo-uit}
\begin{eqnarray}
&&~~~\bu^{(2i)}_{t}=4[g\bu^{(-2)}(S^{(2i,0)}+3e_1S^{(2i,-2)}+gS^{(2i,-4)})+gS^{(2i,-2)}(\bu^{(0)}+3e_1\bu^{(-2)}+g\bu^{(-4)})\nonumber\\
&&~~~~~~~~~~~~-\bu^{(0)}(S^{(2i,2)}+3e_1S^{(2i,0)}+gS^{(2i,-2)})
+S^{(2i,1)}\bu^{(1)}-g\bu^{(-1)}S^{(2i,-1)}+\bu^{(2i+3)}\nonumber\\
&&~~~~~~~~~~~~+3e_1\bu^{(2i+1)}+g\bu^{(2i-1)}-S^{(2i,0)}(\bu^{(2)}+3e_1\bu^{(0)}+g\bu^{(-2)})],\\
&&\bu^{(2i+1)}_{t}=4[g\bu^{(-2)}(S^{(2i+1,0)}+3e_1S^{(2i+1,-2)}+gS^{(2i+1,-4)})-g\bu^{(-1)}S^{(2i+1,-1)}+\bu^{(2i+4)}\nonumber\\
&&~~~~~~~~~~~~-\bu^{(0)}(S^{(2i+1,2)}+3e_1S^{(2i+1,0)}+gS^{(2i+1,-2)})+(9e_1^2+2g)\bu^{(2i)}+g^2\bu^{(2i-4)}
\nonumber\\
&&~~~~~~~~~~~~+6e_1g\bu^{(2i-2)}+6e_1\bu^{(2i+2)}-S^{(2i+1,0)}(\bu^{(2)}+3e_1\bu^{(0)}+g\bu^{(-2)})\nonumber\\
&&~~~~~~~~~~~~+S^{(2i+1,1)}\bu^{(1)}+gS^{(2i+1,-2)}(\bu^{(0)}+3e_1\bu^{(-2)}+g\bu^{(-4)})].
\end{eqnarray}
\end{subequations}
These evolution of $\bu^{(i)}$ can be turned into the evolution of $\S{i,j}$:
\begin{subequations}
\label{evo-Sijx}
\begin{align}
S^{(2i,2j)}_{x}=&S^{(2i+1,2j)}+S^{(2i,2j+1)}-S^{(2i,0)}S^{(0,2j)}+gS^{(2i,-2)}S^{(-2,2j)},  \label{evo-Sijx-1}\\
S^{(2i,2j+1)}_x=&S^{(2i,2j+2)}+3e_1S^{(2i,2j)}+gS^{(2i,2j-2)}-S^{(2i,0)}S^{(0,2j+1)}+S^{(2i+1,2j+1)}\nonumber \\
&+gS^{(2i,-2)}S^{(-2,2j+1)},\\
S^{(2i+1,2j+1)}_x=&S^{(2i+1,2j+2)}+3e_1S^{(2i+1,2j)}+gS^{(2i+1,2j-2)}-S^{(2i+1,0)}S^{(0,2j+1)}+S^{(2i+2,2j+1)}\nonumber \\
&+gS^{(2i+1,-2)}S^{(-2,2j+1)}+3e_1S^{(2i,2j+1)}+gS^{(2i-2,2j+1)},
\end{align}
\end{subequations}
and
\begin{subequations}
\label{evo-Sijt}
\begin{align}
S^{(2i,2j)}_{t}=&4[gS^{(-2,2j)}(S^{(2i,0)}+3e_1S^{(2i,-2)}+gS^{(2i,-4)})+gS^{(2i,-2)}(S^{(0,2j)}+3e_1S^{(-2,2j)}\nonumber\\
&+gS^{(-4,2j)})-S^{(0,2j)}(S^{(2i,2)}+3e_1S^{(2i,0)}+gS^{(2i,-2)})+S^{(2i,1)}S^{(1,2j)}\nonumber\\
&-gS^{(-1,2j)}S^{(2i,-1)}+S^{(2i+3,2j)}+3e_1S^{(2i+1,2j)}+gS^{(2i-1,2j)}-S^{(2i,0)}(S^{(2,2j)}\nonumber\\
&+3e_1S^{(0,2j)}+gS^{(-2,2j)})+S^{(2i,2j+3)}+3e_1S^{(2i,2j+1)}+gS^{(2i,2j-1)}],\\
S^{(2i,2j+1)}_{t}=&4[gS^{(-2,2j+1)}(S^{(2i,0)}+3e_1S^{(2i,-2)}+gS^{(2i,-4)})+gS^{(2i,-2)}(S^{(0,2j+1)}\nonumber\\
&+3e_1S^{(-2,2j+1)}+gS^{(-4,2j+1)})-S^{(0,2j+1)}(S^{(2i,2)}+3e_1S^{(2i,0)}+gS^{(2i,-2)})\nonumber\\
&+S^{(2i,1)}S^{(1,2j+1)}-gS^{(-1,2j+1)}S^{(2i,-1)}+S^{(2i+3,2j+1)}+3e_1S^{(2i+1,2j+1)}\nonumber\\
&+gS^{(2i-1,2j+1)}-S^{(2i,0)}(S^{(2,2j+1)}+3e_1S^{(0,2j+1)}+gS^{(-2,2j+1)})+S^{(2i,2j+4)}\nonumber\\
&+6e_1S^{(2i,2j+2)}+(9e_1^2+2g)S^{(2i,2j)}+6e_1gS^{(2i,2j-2)}+g^2S^{(2i,2j-4)}],
\end{align}
\begin{align}
S^{(2i+1,2j+1)}_{t}=&4[gS^{(-2,2j+1)}(S^{(2i+1,0)}+3e_1S^{(2i+1,-2)}+gS^{(2i+1,-4)})+gS^{(2i+1,-2)}(S^{(0,2j+1)}\nonumber\\
&+3e_1S^{(-2,2j+1)}+gS^{(-4,2j+1)})-S^{(0,2j+1)}(S^{(2i+1,2)}+3e_1S^{(2i+1,0)}+gS^{(2i+1,-2)})\nonumber\\
&+S^{(2i+1,1)}S^{(1,2j+1)}-gS^{(-1,2j+1)}S^{(2i+1,-1)}+S^{(2i+4,2j+1)}+6e_1S^{(2i+2,2j+1)}\nonumber\\
&+6e_1S^{(2i+1,2j+2)}+(9e_1^2+2g)S^{(2i+1,2j)}+6e_1gS^{(2i+1,2j-2)}+g^2S^{(2i+1,2j-4)}\nonumber\\
&-S^{(2i+1,0)}(S^{(2,2j+1)}+3e_1S^{(0,2j+1)}+gS^{(-2,2j+1)})+S^{(2i+1,2j+4)}\nonumber\\
&+(9e_1^2+2g)S^{(2i+1,2j)}+6e_1gS^{(2i+1,2j-2)}+g^2S^{(2i+1,2j-4)}].
\end{align}
\end{subequations}
It is not necessary to write out the $S^{(2i+1,2j)}_x$ and
$S^{(2i+1,2j)}_t$ due to the symmetry property $S^{(i,j)}=S^{(j,i)}$.
One can repeatedly use \eqref{evo-Sijx} and easily get higher-order
$x$-derivatives
of $S^{(i,j)}$  by means of computer algebra, e.g. {\it Mathematica}.

These derivatives of $\S{i,j}$ bring the following epKdV system \cite{NP-JNMP-2003}
\begin{subequations}
\label{c-ell-KdV}
\begin{align}
u_t=&u_{xxx}+6u_x^2-6gs_x^2, \label{c-ell-KdV-1}\\
s_t=&s_{xxx}+6u_x s_x-6gs_x h_x, \label{c-ell-KdV-2} \\
h_t=&h_{xxx}+6s_x^2-6gh_x^2, \label{c-ell-KdV-3}\\
v_t=&v_{xxx}+6v_x u_x+6gs_xS^{(-1,-2)}_x,  \label{c-ell-KdV-4}\\
w_t=&w_{xxx}+6s_xS^{(0,1)}_x-6gw_xh_x, \label{c-ell-KdV-5}
\end{align}
\end{subequations}
where $u,v,s,w,h$ are defined as in \eqref{uvswh}, and
\begin{subequations}
\begin{align}
    S^{(-1,-2)}&=\frac{1}{2}(h_x+s^2-gh^2),\label{S^{(-1,-2)}}\\
S^{(0,1)}&=\frac{1}{2}(u_x+u^2-gs^2)\label{S-01},
\end{align}
\end{subequations}
which are from \eqref{evo-Sijx-1}.

To finish the derivation of \eqref{c-ell-KdV}, one needs long and tedious verification in which
the recurrence relations \eqref{Sij-k1} are successively used.
For example,
for the first equation \eqref{c-ell-KdV-1}, after substitute the expressions of
$u_t, u_{xxx}, u_x, s_x$ one has
\begin{align*}
&\frac{1}{6}(u_t-u_{xxx}-6u_x^2+6gs_x^2)\\
=~&g^2 S^{(0, -2)} [S^{(0, -4)} - S^{(-1, -2)} S^{(0, -2)} +
      S^{(-2, -2)} (-1 + S^{(0, -1)})] \\
      & +
   g [-
      S^{(0, -2)} (-S^{(0, 0)} (1 + S^{(1, -2)}) +S^{(2, -2)} +{S^{(0, -2)}} S^{(1, 0)}) \\
   &-S^{(1, -2)} + S^{(0, -1)} (1 + S^{(1, -2)})-
         S^{(1, -1)}S^{(0, -2)} ] \\
        & - {S^{(1, 0)}}^2 +
   S^{(0, 0)} S^{(1, 1)}- S^{(2, 1)} + S^{(3, 0)}.
\end{align*}
It vanishes in light of the relations \eqref{Sij-k1} with
$(i,j)=(0,-2)$ and $(0,1)$.
Equations \eqref{c-ell-KdV-2}-\eqref{c-ell-KdV-5} can rigorously  be derived in a similar way
but the procedures are long. Here we skip them.

In \cite{NP-JNMP-2003}
\begin{equation}
A=-u+\frac{w}{s}
\end{equation}
is introduced, by which the epKdV system \eqref{c-ell-KdV} yielded the following coupled equation, (i.e. \eqref{c-ell-KdV-couple-sec1})
\begin{subequations}
\begin{align}
s_t=&4s_{xxx}+6s_x[R(s^2)-A^2-\frac{2A s_x}{s}-\frac{2s_{xx}}{s}], \label{new-st}\\
A_t=&4A_{xxx}-6A^2A_x +6A_xR(s^2)-\frac{6 s_x}{s}(R(s^2))_x, \label{new-At}
\end{align}
\label{eq:As}
\end{subequations}
with the elliptic curve $R(x)$ given in \eqref{ell-curve-2}. This coupled system is
integrable in the sense of admitting a continuous Lax
pair\cite{NP-JNMP-2003} (also see Sec.\ref{sec-5-3}).

\subsection{Lax pair }\label{sec-5-3}

Let us consider \eqref{evo-uix} with  $i=0$, which reads
\begin{subequations}\label{evo-uix-i=0-1}
\begin{align}
\bu^{(0)}_x&=\bu^{(1)}-u\bu^{(0)}+gs\bu^{(-2)},\\
\bu^{(1)}_x&=\bu^{(2)}-S^{(1,0)}\bu^{(0)}+gw\bu^{(-2)}+3e_1\bu^{(0)}.
\end{align}
\end{subequations}
After replacing $\bu^{(-2)},~\bu^{(2)}$ and $S^{(1,0)}$ with \eqref{u-2}, \eqref{u+2} and \eqref{S-01}, respectively, we have
\begin{subequations}\label{evo-uix-i=0-2}
\begin{align}
\bu^{(0)}_x&=\bu^{(1)}-u\bu^{(0)}+gs\bK^{-1}(w{\bu}^{(0)}-s{\bu}^{(1)}),\label{evo-uix-i=0-2-a}\\
\bu^{(1)}_x&=(\bK+gw^2\bK^{-1}){\bu}^{(0)}+(3e_1-u_x-u^2+gs^2)\bu^{(0)}+(u\bI -gsw\bK^{-1}){\bu}^{(1)}.\label{evo-uix-i=0-2-b}
\end{align}
\end{subequations}
As in the discrete case, we can consider   $\bk$ to be its canonical form and then from the first rows of
\eqref{evo-uix-i=0-2-a} and \eqref{evo-uix-i=0-2-b} we find the linear form
 \begin{subequations}\label{lax-con}
 \begin{equation}\label{lax-con-x}
   \phi_x=\left(
          \begin{array}{cc}
            - u+\frac{g}{K}sw & 1- \frac{g}{K}s^2\\
            K+3e+gs^2-u^2-u_x+\frac{g}{K}w^2 & u-\frac{g}{K}ws \\
          \end{array}
        \right)\phi,
 \end{equation}
where $\phi$ is defined by \eqref{phi}.
In a similar way from \eqref{evo-uit} we can find the time evolution of $ \phi$, which is formulated by
\begin{equation}\label{lax-con-t}
   \phi_t=-\left(
          \begin{array}{cc}
            S^{(0,1)}_x & u_x\\
            S^{(1,1)}_x & S^{(0,1)}_x  \\
          \end{array}
        \right)\phi-\frac{g}{K}\left(
          \begin{array}{cc}
            (1-vw)\frac{s_x}{s}+v_xw & vs_x-sv_x\\
            (1-vw)\frac{w_x}{s} -w(\frac{1-vw}{s})_x& -(1-vw)\frac{s_x}{s}-v_xw  \\
          \end{array}
        \right)\phi,
 \end{equation}
 \end{subequations}
 where
 \begin{equation}\label{S^{(1,1)}_x}
   S^{(1,1)}_x=\frac{1}{2}S^{(0,1)}_{xx}+uS^{(0,1)}_x-gsw_x,
 \end{equation}
and  $S^{(0,1)}$ is given by \eqref{S-01}.

\eqref{lax-con} can be viewed as a Lax pair of the system \eqref{c-ell-KdV},
which can also be derived from the direct linearization approach  \cite{NP-JNMP-2003}.
The compatibility gives equations \eqref{c-ell-KdV-1}, \eqref{c-ell-KdV-2} and \eqref{c-ell-KdV-5}.

\section{Straight continuum limits}\label{sec-6}

The skew continuum of the elpKdV system \eqref{ell-la} was considered in \cite{NP-JNMP-2003} when they studied initial value problems of  \eqref{ell-la}.
Such a limit is performed by
introducing  the skew-change of variables $(n,m)\mapsto (\mathcal{N}=n+m,m)$.

Let us consider the straight continuum limit, where we first take
\begin{equation}
m \rightarrow \infty,~~ b\rightarrow \infty,~~ {\rm while}~
\frac{m}{b}=\tau-\tau_0 \sim O(1),
\label{str-lim-1-H1}
\end{equation}
with $\tau_0$ being a constant.
Define
\begin{equation}
    u=u_{n,m}=:u_n(\tau),~~s=s_{n,m}=:s_n(\tau),~~w=w_{n,m}=:w_n(\tau).
\end{equation}
Then,  applying the Taylor expansions  into \eqref{ell-la} at $\tau$,
the leading term (in terms of $1/b$) of each equation yields the following semi-discrete equations
\begin{subequations}\label{semi-dis-str}
\begin{align}
&\partial_{\tau}(\t u+u)=2a(\t u-u)- (\t u-u)^2+g(\t s-s)^2,\label{semi-dis-str-1}\\
&\partial_{\tau}(s \t s)=(\t s-s)(a\t s+as-\t w+w)+ u\t s^2+\t us^2-s\t s(u+\t u),\label{semi-dis-str-2}\\
& (a+ u-\frac{\t{w}}{\t{s}})(a-\t{u}+\frac{w}{s})=a^2-R(s\t{s}),\label{semi-dis-str-4}\\
&\partial_{\tau}(u+\frac{w}{s})+(u-\frac{w}{s})^2=R(s^2),\label{semi-dis-str-5}
\end{align}
\end{subequations}
and we note both \eqref{ell-la-2} and \eqref{ell-la-3} yield \eqref{semi-dis-str-2} in the continuum limit.

For the full limit of   \eqref{semi-dis-str}, first we take
\begin{equation}
n \rightarrow \infty,~~ a \rightarrow \infty,~~ {\rm while}~
\frac{n}{a}=\xi \sim O(a^2)
\label{str-lim-2-H1}
\end{equation}
and then introduce continuous variables $x$ and $t$ by
\begin{equation}
x=\tau+\xi,~~~ t=\frac{\xi}{12a^2}, \label{xt-skew-lim-H1}
\end{equation}
with $\xi$ as an auxiliary variable.
Then, under the coordinates $(x,t)$ both \eqref{semi-dis-str-4} and \eqref{semi-dis-str-5} yield
\begin{equation}\label{last-con}
    \Bigl(u+\frac{w}{s}\Bigr)_x+\Bigl(u-\frac{w}{s}\Bigr)^2=R(s^2);
\end{equation}
\eqref{semi-dis-str-1} yields \eqref{c-ell-KdV-1};
\eqref{semi-dis-str-2} gives \eqref{c-ell-KdV-2}
and for that we need to make use of a relation $s_{xx}=2(g sh_x-us_x+w_x)$,
which is obtained as a continuous limit of the summation of $\t {\eqref{re-la-41}}$ and \eqref{re-la-4}.

If we employ the transformation $A=-u+\frac{w}{s}$ in the equations \eqref{semi-dis-str-2} and \eqref{semi-dis-str-4},
it turns out to be
\begin{subequations}
\begin{align}
&\partial_{\tau}(s \t s)=(\t s-s)(a\t s+as-\t A \t s+As)+ u\t s^2+\t us^2-s\t s(u+\t u),\label{semi-dis-str-31}\\
&(a+ u-\t
A-\t{u})(a-\t{u}+A+u)=a^2-R(s\t{s}).\label{semi-dis-str-41}
\end{align}
\end{subequations}
Their continuum limit gives the coupled system \eqref{eq:As},
where we need to use  the relation $ u_x=\frac{1}{2}(R(s^2)-A^2-A_x)$,
which is \eqref{last-con} in terms of $A$ and $u$, derived from \eqref{semi-dis-str-5}.

\section{Conclusions}\label{sec-7}

In this paper a new class of solutions of the elliptic KdV systems (both the discrete \eqref{ell-la-kdv}
and the continuous \eqref{c-ell-KdV-couple-sec1})
has been uncovered.
We made use of Sylvester-type equation with elliptic ingredient.
Solutions can be classified by the canonical form of $\bk$, which are much richer than pure solitons.

A Cauchy matrix  dressed by dispersion relations usually satisfies a Sylvester equation.
Starting from the Sylvester equation and dispersion relations,
not only integrable equations can be derived but also their solutions and Lax pairs can be constructed.
Such the Cauchy matrix approach is particularly powerful in the study of
discrete integrable systems (see \cite{NAH-2009-JPA,ZZ-SAM-2013}), as well as continuous systems \cite{XZZ-2014-JNMP}.
Dressed Cauchy matrices also play key roles in the so-called operator method \cite{M-book-1987,AC-1996-JMP,S-PD-1998,CS-Non-1999,CS-DMV-2000,S-LAA-2010},
trace method \cite{B-Non-2000}, etc.

Here we specially mention the connection with Inverse Scattering Transform (IST),
which is to express solutions of the Gel¡¯fand-Levitan-Marchenko equation
via a triplet $(\bA;\mathbf{b};\mathbf{c})$ where matrix $\bA$ and vectors $\mathbf{b}$ and $\mathbf{c}$
satisfy some Sylvester equations\cite{AM-IP-2006,ADM-IP-2007,ADM-JMP-2010}.
Solutions are classified by the canonical form of $\bA$.
The obtained breather-like solutions can describe potentials on half-line and with non-zero refection coefficients \cite{AM-IP-2006}.
We believe such connections can be found for the fully discrete systems which were solved via IST recently \cite{BJ-IP-2010,Butler-Non-2012}.
Besides, since there is deep connection between Cauchy matrix approach and direct linearization scheme,
it is possible to develop a more general direct linearization scheme.

Adding elliptic information to integrable systems, either to equations or to solutions,
is an interesting topic, which brings the study of the integrable systems into a larger area and more insight
\cite{NA-IMRN-2010,AN-CMP-2010,CX-JPA-2012-1,CZ-JPA-2012,N-2013-Leiden,YN-JMP-2013,JN-JMP-2014,DNY-JPA-2015}.
In this paper, we have shown that the Cauchy matrix approach works for the study of some elliptic integrable systems,
i.e. some terms of these systems are formulated with  an elliptic curve.
The Sylvester equation \eqref{SE-2a} is our starting point.
We derived the discrete as well as continuous elliptic KdV systems.
Apart from solutions, we also obtained Lax representations.
As for solutions, the discrete plane wave factor \eqref{rho-i-d} and continuous one \eqref{rho-i-c}
are defined with the wave numbers $k_i$ which together with $K_i$ obeys the elliptic curve \eqref{ell-curve-2}.
For the Lax pairs  \eqref{dlax pair} and \eqref{lax-con}, $(k,K)$ plays the role of spectral parameters
which also obeys the elliptic curve.

\vskip 20pt
\subsection*{Acknowledgments}
YYS and DJZ are partly supported by the NSFC (No.11371241) and the Project of ``First-class Discipline of Universities in Shanghai''.
FWN is partly supported by the EPSRC grant EP/I038683/1.

\vskip 10pt

\begin{appendix}

\section{Lower triangular Toeplitz matrices}\label{A:1}

Here we collect some properties of Lower triangular Toeplitz (LTT) matrices.

A $N$-th order LTT matrix is a matrix of the following form
\begin{equation}
\bT^{\tyb{N}}(\{a_j\}^{N}_{1})
=\left(\begin{array}{cccccc}
a_1 & 0    & 0   & \cdots & 0   & 0 \\
a_2 & a_1  & 0   & \cdots & 0   & 0 \\
a_3 & a_2  & a_1 & \cdots & 0   & 0 \\
\vdots &\vdots &\cdots &\vdots &\vdots &\vdots \\
a_{N} & a_{N-1} & a_{N-2}  & \cdots &  a_2   & a_1
\end{array}\right)_{N\times N}.
\label{LTT}
\end{equation}
Let
\begin{equation}
\mathcal{T}^{\tyb{N}}=\{\bT^{\tyb{N}}(\{a_j\}^{N}_{1})\},
\end{equation}
then we have $\bA \bB=\bB \bA,~\forall \bA, \bB\in \mathcal{T}^{\tyb{N}}$, i.e.
$\mathcal{T}^{\tyb{N}}$ is a commutative set w.r.t. matrix multiplication.
Particularly, the subset
\begin{equation}
\mathcal{T}_1^{\tyb{N}}=\{\bF \in \mathcal{T}^{\tyb{N}}\, | \, \mathrm{det}(\bF)\neq 0\}
\end{equation}
is an Abelian group.

Obviously, the Jordan block matrix
\begin{equation}
\Ga^{\tyb{N}}_{\ty{J}}(a)
=\left(\begin{array}{cccccc}
a & 0    & 0   & \cdots & 0   & 0 \\
1   & a  & 0   & \cdots & 0   & 0 \\
0   & 1  & a   & \cdots & 0   & 0 \\
\vdots &\vdots &\vdots &\vdots &\vdots &\vdots \\
0   & 0    & 0   & \cdots & 1   & a
\end{array}\right)
\label{Jord-A-1}
\end{equation}
is a LTT matrix, and one can verify the following.
\begin{Proposition}
\label{prop-A-1}
If $\bA \in \mathbb{C}_{N\times N}$ and $\Ga^{\tyb{N}}_{\ty{J}}(a)\bA=\bA \,\Ga^{\tyb{N}}_{\ty{J}}(a)$, then there must be
$\bA\in \mathcal{T}^{\tyb{N}}$.
\end{Proposition}

If $a_j\in \mathbb{C}$, then the LTT matrix \eqref{LTT} can be generated by certain function.
Suppose that $f(k)$ is an analytic function. Using Taylor coefficients
\begin{equation}
a_j=\frac{1}{(j-1)!}\partial_k^{j-1}f(k)|_{k=k_0},~~j=1,2,\cdots,N
\label{aj-f(k)-A}
\end{equation}
we can generate a LTT matrix. The matrix \eqref{LTT} with \eqref{aj-f(k)-A} is called a LTT matrix generated by $f(k)$ at $k=k_0$,
denoted by $\bT^{\tyb{N}}[f(k_0)]$, and $f(k)$ is called generating function.
In this sense, the Jordan block  \eqref{Jord-A-1} is generated by $f(k)=k$ at $k=a$,
and the unit matrix $\bI$ is generated by $f(k)\equiv 1$.
On the other hand, for any LTT matrix \eqref{LTT} with $a_j\in \mathbb{C}$, it can be generated by
the polynomial
\begin{equation}
\alpha(k)=\sum^N_{j=1} a_j (k-k_0)^{j-1}
\label{alpha-k}
\end{equation}
with $\{a_j\}$ as coefficients.
Next, by $[f(k)]^{\tyb{N}}_{k_0}$ we denote a set of functions (equivalence class)
in which all the functions have the same $(N-1)$-th order Taylor polynomial at $k=k_0$ as $f(k)$ has.
Say, $f(k) \sim g(k)$ if they have same $(N-1)$-th order Taylor polynomial at $k=k_0$.
Thus, the LTT matrix \eqref{LTT} can be generated by any $f(k)\in [\alpha(k)]^{\tyb{N}}_{k_0}$.
With such correspondence, we have the following.

\begin{Proposition}
\label{prop-A-3b}
If $\bA=\bT^{\tyb{N}}[f(k_0)]$ and $\bB=\bT^{\tyb{N}}[g(k_0)]$, then
\begin{equation}
\bC=\bA\bB=\bT^{\tyb{N}}[f(k_0)g(k_0)].
\label{CAB}
\end{equation}
i.e., $\bA\bB$ is a LTT matrix generated by $f(k)g(k)$  at $k_0$.
As a result, we have
\[\prod^{s}_{j=1}\bT^{\tyb{N}}[f_j(k_0)]=\bT^{\tyb{N}}[\prod^{s}_{j=1}f_j(k_0)],\]
and
\begin{equation*}
(T^{\tyb{N}}[f(k_0)])^{-1}=\bT^{\tyb{N}}[1/f(k_0)]
\end{equation*}
if $f(k_0)\neq 0$.
\end{Proposition}

\begin{proof}
We only need to prove \eqref{CAB}.
Suppose
\[\bA=(a_{ij})_{N\times N},~~\bB=(b_{ij})_{N\times N},~~\bC=(c_{ij})_{N\times N}.\]
Then we have (with $k=k_0$)
\[a_{ij}=\left\{\begin{array}{ll}
                \frac{1}{(i-j)!}\partial_{k}^{i-j} f(k), & i\geq j,\\
                0, & i < j,
                \end{array}\right.,~~
b_{ij}=\left\{\begin{array}{ll}
                \frac{1}{(i-j)!}\partial_{k}^{i-j} g(k), & i\geq j,\\
                0, & i < j,
                \end{array}\right.
\]
and (with $k=k_0$)
\begin{align*}
c_{ij}& =\sum^N_{s=1}a_{is}b_{sj}\\
      & =\sum^N_{s=1}  \frac{1}{(i-s)!}\partial_{k}^{i-s} f(k)\cdot \frac{1}{(s-j)!}\partial_{k}^{s-j} g(k)\\
      & = \sum^N_{l=i-j}  \frac{1}{(i-j-l)!\, l!}\partial_{k}^{i-j-l} f(k)\cdot \partial_{k}^{l} g(k), ~~(i\geq j),\\
      & =  \frac{1}{(i-j)!}\sum^N_{l=i-j}  \partial_{k}^{i-j} (f(k) g(k)), ~~(i\geq j),
\end{align*}
and $c_{ij}=0$ when $i<j$.
Thus, \eqref{CAB} is proved.
\end{proof}

In addition to the LTT matrices, we define the following skew triangular Toeplitz (STT) matrix:
\begin{equation}
\bH^{\tyb{N}}(\{b_j\}^{N}_{1})
=\left(\begin{array}{ccccc}
b_1 & \cdots  & b_{N-2}  & b_{N-1} & b_N\\
b_2 & \cdots & b_{N-1}  & b_N & 0\\
b_3 &\cdots & b_N & 0 & 0\\
\vdots &\vdots & \vdots & \vdots & \vdots\\
b_N & \cdots & 0 & 0 & 0
\end{array}
\right)_{N\times N},
\label{STT}
\end{equation}

The following property holds \cite{XZZ-2014-JNMP}.
\begin{Proposition}\label{prop-A-4}
Let
\begin{align}
& \bar{\mathcal{T}}^{\tyb{N}}=\{\bH^{\tyb{N}}(\{b_j\}^{N}_{1})\}.
\end{align}
Then we have\\
\textrm{(1).}~ $\bH=\bH^T,~\forall \bH\in \bar{\mathcal{T}}^{\tyb{N}}$;\\
\textrm{(2).}~ $\bH \bA=(\bH \bA)^T=\bA^T \bH,~\forall \bA\in \mathcal{T}^{\tyb{N}},~\forall \bH\in \bar{\mathcal{T}}^{\tyb{N}}$.
\end{Proposition}

It can be extended to the following generic case.
\begin{Proposition}\label{prop-A-5}
Let
\begin{subequations}
\begin{align}
& \mathcal{G}^{\tyb{N}}=\{\mathrm{Diag}(\Ga^{\tyb{N}}_{\ty{D}}(\{a_{1,j}\}^{N_1}_{1}), \bT^{\tyb{N}}(\{a_{2,j}\}^{N_2}_{1}),
\bT^{\tyb{N}}(\{a_{3,j}\}^{N_3}_{1}), \cdots , \bT^{\tyb{N}}(\{a_{s,j}\}^{N_s}_{1}))\},\\
& \bar{\mathcal{G}}^{\tyb{N}}=\{\mathrm{Diag}(\Ga^{\tyb{N}}_{\ty{D}}(\{b_{1,j}\}^{N_1}_{1}), \bH^{\tyb{N}}(\{b_{2,j}\}^{N_2}_{1}),
\bH^{\tyb{N}}(\{b_{3,j}\}^{N_3}_{1}), \cdots , \bH^{\tyb{N}}(\{b_{s,j}\}^{N_s}_{1}))\},
\end{align}
where $0\leq N_j\leq N$ for $j=0,1,\cdots,N$ and $\sum^{s}_{j=1}N_j=N$.
\end{subequations}
Then we have\\
\textrm{(1).}~ $\bA \bB=\bB \bA,~\forall \bA, \bB\in \mathcal{G}^{\tyb{N}}$;\\
\textrm{(2).}~ $\bH=\bH^T,~\forall \bH\in \bar{\mathcal{G}}^{\tyb{N}}$;\\
\textrm{(3).}~ $\bH \bA=(\bH \bA)^T=\bA^T \bH,~\forall \bA\in \mathcal{G}^{\tyb{N}},~\forall \bH\in \bar{\mathcal{G}}^{\tyb{N}}$.
\end{Proposition}

\section{Explicit forms of $\br$ and $\bc$}\label{A:2}

Here we list out the explicit forms of $\br$ and $\bc$ satisfying \eqref{dr} and \eqref{evo-rs}, respectively.

\subsection{Solution to \eqref{dr}}

\noindent
(1).~ When $\Ga=\Ga^{\tyb{N}}_{\ty{D}}(\{k_j\}^{N}_{1})$ we have
\begin{align}
\br=\br_{\hbox{\tiny{\it D}}}^{\hbox{\tiny{[{\it N}]}}}(\{k_j\}_{1}^{N})=(r_1, r_2, \cdots, r_N)^T, ~~\mathrm{with}~ r_i=\rho_i,
\label{r-d}
\end{align}
where
\begin{equation}
\rho_i=\Bigl(\frac{a+k_i}{a-k_i}\Bigr)^n\Bigl(\frac{b+k_i}{b-k_i}\Bigr)^m \rho_{i}^0,
\label{rho-i-d}
\end{equation}
and $\rho_{i}^0$ is a constant.\\
(2).~When $\Ga=\Ga^{\tyb{N}}_{\ty{J}}(k_1)$, we have
\begin{align}
\br=\br_{\hbox{\tiny{\it J}}}^{\hbox{\tiny{[{\it N}]}}}(k_1)=(r_1, r_2, \cdots, r_N)^T, ~~\mathrm{with}~ r_i=\frac{\partial^{i-1}_{k_1}\rho_1}{(i-1)!},
\label{r-jor}
\end{align}
where $\rho_1$ is defined in \eqref{rho-i-d}.\\
(3).~When $\Ga=\Ga^{\tyb{N}}_{\ty{G}}$ we have
\begin{equation}
\br=\left(
\begin{array}{l}
\br_{\ty{D}}^{\tyb{N$_1$}}(\{k_j\}_{1}^{N_1})\\
\br_{\ty{J}}^{\tyb{N$_2$}}(k_{N_1+1})\\
\br_{\ty{J}}^{\tyb{N$_3$}}(k_{N_1+2})\\
\vdots\\
\br_{\ty{J}}^{\tyb{N$_s$}}(k_{N_1+(s-1)})
\end{array}
\right),
\label{r-d-g}
\end{equation}
where $\br_{\hbox{\tiny{\it D}}}^{\hbox{\tiny{[{\it N}$_1$]}}}(\{k_j\}_{1}^{N_1})$ and
$\br_{\hbox{\tiny{\it J}}}^{\hbox{\tiny{[{\it N}$_i$]}}}(k_j)$ are defined as in \eqref{r-d} and  \eqref{r-jor}, respectively.

\subsection{Solution to \eqref{evo-rs}}

\noindent
(1).~ When $\Ga=\Ga^{\tyb{N}}_{\ty{D}}(\{k_j\}^{N}_{1})$ we have
$\br=\br_{\hbox{\tiny{\it D}}}^{\hbox{\tiny{[{\it N}]}}}(\{k_j\}_{1}^{N})$ is given in the form \eqref{r-d} and
\begin{align}
\bc=\bc_{\hbox{\tiny{\it D}}}^{\hbox{\tiny{[{\it N}]}}}(\{k_j\}_{1}^{N})=(c_1, c_2, \cdots, c_N)^T, ~~\mathrm{with}~ c_i=r_i,
\label{c-c}
\end{align}
but here
\begin{align}
\rho_i=e^{\xi_i},~~\xi_i=k_{i}x+ 4k^3_i t +\xi^{(0)}_i,~\mathrm{with~ constant~}\xi^{(0)}_i.
\label{rho-i-c}
\end{align}
(2).~When $\Ga=\Ga^{\tyb{N}}_{\ty{J}}(k_1)$, we have
$\br=\br_{\hbox{\tiny{\it J}}}^{\hbox{\tiny{[{\it N}]}}}(k_1)$ is given in the form \eqref{r-d} and
\begin{align}
\bc=\bc_{\hbox{\tiny{\it J}}}^{\hbox{\tiny{[{\it N}]}}}(k_1)=(c_1, c_2, \cdots, c_N)^T, ~~\mathrm{with}~ c_i=r_{N-i},
\label{c-jor}
\end{align}
where $\rho_1$ is defined in \eqref{rho-i-c}.\\
(3).~When $\Ga=\Ga^{\tyb{N}}_{\ty{G}}$ we have $\br$ is given in the form of \eqref{r-d-g} and
\begin{equation}
\bc=\left(
\begin{array}{l}
\bc_{\ty{D}}^{\tyb{N$_1$}}(\{k_j\}_{1}^{N_1})\\
\bc_{\ty{J}}^{\tyb{N$_2$}}(k_{N_1+1})\\
\bc_{\ty{J}}^{\tyb{N$_3$}}(k_{N_1+2})\\
\vdots\\
\bc_{\ty{J}}^{\tyb{N$_s$}}(k_{N_1+(s-1)})
\end{array}
\right),
\end{equation}
where $\bc_{\hbox{\tiny{\it D}}}^{\hbox{\tiny{[{\it N}$_1$]}}}(\{k_j\}_{1}^{N_1})$ and
$\bc_{\hbox{\tiny{\it J}}}^{\hbox{\tiny{[{\it N}$_i$]}}}(k_j)$ are defined as in \eqref{c-c} and  \eqref{c-jor}, respectively.

\end{appendix}

\end{document}